\documentclass[11pt,letterpaper]{article}

\usepackage{todonotes}

\usepackage{subfig}

\usepackage{enumitem}

\usepackage[margin=1in]{geometry}
\usepackage{amsmath, amssymb, amsthm, thmtools, amsfonts}
\usepackage[utf8]{inputenc}
\usepackage{appendix}
\usepackage{xspace}
\usepackage{color}
\usepackage[ruled,linesnumbered]{algorithm2e}
\usepackage[noend]{algpseudocode}
\definecolor{darkred}{rgb}{0.5,0,0}
\definecolor{darkgreen}{rgb}{0,0.5,0}
\definecolor{darkblue}{rgb}{0,0,0.5}
\usepackage{hyperref}
\hypersetup{
    unicode=false,          
    colorlinks=true,        
    linkcolor=darkred,          
    citecolor=darkgreen,        
    filecolor=magenta,      
    urlcolor=cyan           
}

\usepackage[capitalize, nameinlink]{cleveref}
\newtheorem{theorem}{Theorem}[section]
\newtheorem{lemma}[theorem]{Lemma}

\newtheorem{corollary}[theorem]{Corollary}
\newtheorem{proposition}[theorem]{Proposition}
\newtheorem{remark}[theorem]{Remark}
\newtheorem{definition}[theorem]{Definition}
\newtheorem{observation}[theorem]{Observation}
\newtheorem{claim}[theorem]{Claim}

\Crefname{theorem}{Theorem}{Theorems}
\Crefname{lemma}{Lemma}{Lemmas}
\Crefname{claim}{Claim}{Claims}
\Crefname{fact}{Fact}{Facts}
\Crefname{corollary}{Corollary}{Corollaries}
\Crefname{proposition}{Proposition}{Propositions}

\newcommand{\bigO}{\mathcal{O}}
\newcommand\eps{\varepsilon}
\renewcommand{\paragraph}[1]{\vspace{0.15cm}\noindent {\bf #1}:}

\mathchardef\mhyphen="2D

\newcommand{\poly}{\operatorname{\text{{\rm poly}}}}
\newcommand{\polylog}{\operatorname{\text{{\rm polylog}}}}

\newcommand{\depend}{depend }
\newcommand{\bad}{bad }
\newcommand{\unfortunate}{unfortunate }
\newcommand{\unmotivated}{unmotivated }
\newcommand{\necessary}{necessary }

\newcommand{\set}[1]{\left\{#1\right\}}
\newcommand{\paren}[1]{\mathopen{}\left(#1\right)\mathclose{}}

\newcommand{\ceil}[1]{\mathopen{}\left\lceil#1\right\rceil\mathclose{}}

\newcommand{\card}[1]{\left|#1\right|}

\newcommand{\betterColouringPaletteSize}{\mathopen{} O(\min\{\alpha \log \alpha , \alpha \log \log \log n \}) \mathclose{}}

\newlength{\commentWidth}
\newcommand{\atcp}[1]{\tcp*[r]{\makebox[\commentWidth]{#1\hfill}}}

\newcommand{\knocomment}[1]{{}}
\newcommand{\abgccomment}[1]{{}}
\newcommand{\ercomment}[1]{{}}

\usepackage{authblk}

\begin{document}
\thispagestyle{empty}

\DontPrintSemicolon 

\date{}

\title{Improved Dynamic Colouring of Sparse Graphs}


\author[1]{Aleksander B. G. Christiansen}
\author[2]{Krzysztof D. Nowicki}
\author[1]{Eva Rotenberg}
\affil[1]{Technical University of Denmark\thanks{Partially supported by the VILLUM Foundation grant 37507 ``Efficient Recomputations for Changeful Problems''.}}
\affil[2]{University of Copenhagen\thanks{Partially supported by the VILLUM Foundation grant 16582, ``BARC".}; pathway.com}

\setcounter{page}{0}

\maketitle

\begin{abstract}
Given a dynamic graph subject to edge insertions and deletions, we show how to update an implicit representation of a proper vertex colouring, such that colours of vertices are computable upon query time. We give a deterministic algorithm that uses $\mathcal{O}(\alpha ^2)$ colours for a dynamic graph of arboricity $\alpha$, and a randomised algorithms that uses $O(\min\{\alpha \log \alpha, \alpha \log \log \log n\})$ colours in the oblivious adversary model. Our deterministic algorithm has update- and query times polynomial in $\alpha$ and $\log n$, and our randomised algorithm has amortised update- and query time that with high probability is polynomial in $\log n$ with no dependency on the arboricity. 

Thus, we improve the number of colours exponentially compared to the state-of-the art for implicit colouring, namely from $O(2^\alpha)$ colours, and we approach the theoretical lower bound of $\Omega(\alpha)$ for this arboricity-parameterised approach. Simultaneously, our randomised algorithm improves the update- and query time to run in time solely polynomial in $\log n$ with no dependency on $\alpha$. Our algorithms are fully adaptive to the current value of the dynamic arboricity at query or update time.

\end{abstract}

\thispagestyle{empty}

\newpage

\section{Introduction and Related Work}
In the vertex colouring problem we are asked to assign colours to the vertices of a graph in such a way that no neighbouring vertices are assigned the same colour. More precisely, we are given an undirected and unweighted $n$-vertex graph G = (V,E), in which the vertices have distinct labels from some range polynomial  in $n$
; the goal is to assign to each vertex $v \in V$ a colour $C(v)$ from set $\set{1, 2, \dots, c}$, in a way that does not assign the same colour to any pair of neighbours, i.e. $\forall u,v \in V$, $uv \in E \Rightarrow C(u) \neq C(v)$. The smallest number of colours one needs to vertex colour a graph $G$ is called the \emph{chromatic number} of $G$. Vertex colouring is one of the basic graph problems that has been studied in one way or another for over 150 years \footnote{The question whether every planar graph can be coloured with 4 colours was raised already in 1852; in 1879, it was brought up by Arthur Cayley at a meeting of the London Mathematical Society.}. 

Aside from being a fundamental and well-motivated graph problem, colouring also has ties to other important graph problems such as computing matchings and independent sets, in the sense that techniques for one of these problems not seldomly have implications for the others, especially in modern computing models such as distributed or dynamic algorithms.


\paragraph{Upper bounds on chromatic number}
Given a graph of max degree $\Delta$, one can colour it with $\Delta+1$ colours quite simply. Indeed, consider the vertices in an arbitrary order and assign to each vertex one of the free colours that were not assigned to the neighbours earlier in the ordering. In 1940, Brooks observed that the only graphs that need $\Delta+1$ colours are cliques and odd cycles \cite{brooks_1941}. This hints to the fact that the maximum degree of a graph might not necessarily be a good parameter, when it comes to colourings. In particular, trees can be 2-coloured, but they can have arbitrarily high maximum degree. In this paper, we consider colourings that use palettes of size that depends on the \emph{arboricity} of the input graph, to which we sometimes refer as \emph{arboricity dependent colourings}. 

In a natural language, the arboricity of a graph is the smallest number of forests one needs to use to cover all the edges of a graph. Note that arboricity $\alpha(G) \leq \Delta(G)$ always. It is known, that graphs of arboricity $\alpha$ can be coloured with $2\alpha$ colours. 

As such, arboricity seems to be slightly better parameter to consider than the maximum degree in the context of vertex colouring, as it captures the fact that sparse graphs can be coloured with few colours.
Still, it is still not perfect, as there are graphs with low chromatic number and high arboricity (e.g., $K_{n,n}$). However, in general, determining -- or even approximating -- the chromatic number of a graph is an NP-hard problem~\cite{10.1145/1132516.1132612,DBLP:conf/icalp/KhotP06}, as opposed to determining the maximum degree or approximating the arboricity~\cite{GabowWestermann,Gabow1995,blumenstock2019constructive}. 

\paragraph{Vertex colouring in the LOCAL model} 
The LOCAL model of computation was first introduced by Linial in \cite{DBLP:journals/siamcomp/Linial92}. In this model, we focus mostly on graph problems where the input graph is viewed as a communication network, and the computation is performed by the vertices of the network in synchronous rounds. In each round, each pair of neighbours in the network is allowed to exchange messages of unbounded size, and perform unbounded local computation. As such, any $R$-round LOCAL algorithms is essentially a function that map $R$-hop neighbourhood of a vertex to a label that is part of the collective output. 

Vertex colouring is one of the central problems in the area of distributed computing, and it has been widely studied in this model of computation~\cite{DBLP:journals/siamcomp/Linial92,SzegedyV93,BarenboimEK14,Barenboim16,HarrisSS16,ChangLP18,HalldorssonKMT21,Balliu0KO22}.

\paragraph{Dynamic graphs} In this paper, we consider the vertex colouring problem of low arboricity graphs in the dynamic setting, in which graphs are subject to local changes in the form of edge-insertions and edge-deletions.
In many practical applications, one is working with a sequence of slowly evolving data sets, rather than with a set of independent inputs. Thus, if even linear time algorithms are too slow to be feasible, we can yet devise algorithms that adapt to local changes faster than recomputing everything from scratch. Such algorithms are called \emph{dynamic (graph) algorithms}.

In this paper, we consider a variant of dynamic graph algorithm that allows edge updates. That is, each update can either insert an edge into the graph or delete an edge from the graph. The input of a dynamic algorithm is a sequence of updates $u_1, u_2, \dots$, and as a result, we require that the algorithm maintains some access to the output. This access can be given in multiple ways.
The first way of providing access to the output is to maintain it explicitly. That is, the output to the problem is stored explicitly, and for each update $u_i$ the algorithm outputs a sequence of updates $o_{(i,1)}, o_{(i,2)}, \dots$ to the output. The second way of providing output is via a query data structure. That is, after each update $u_i$, the algorithm may receive a sequence of queries, one by one, $q_{(i,1)}, q_{(i,2)}, \dots$, for each of which it needs to provide an answer. The format of a query depends on the problem, and sometimes even for one problem one might consider several types of queries. E.g., for the $c$-colouring problem the most natural query consists of one vertex, and the answer to such a query is the colour from the maintained colouring, $\{1,2,\ldots,c\}$. If multiple vertices are queried, the partial colouring returned must be extendable to a proper $c$-colouring of all vertices.

Similar \emph{implicit} models are often necessary for dynamic graphs where the recourse after an update is higher than the desired update time. For example, when maintaining two-edge connectivity~\cite{HolmLT98,Thorup00,HolmRT18} or a planar embedding of a graph~\cite{Poutre94,HolmR20soda,HolmR20stoc}; after just one update, the number of two-edge-connected components may change linearly, or linearly many vertices must change coordinates in the plane. In these cases, the state-of-the-art algorithms facilitate queries about relationships between vertices or rotational systems of edges, and an implicit representation of the structure facilitates such queries.

In the case of arboricity dependent vertex colouring, by the results of \cite{DBLP:journals/algorithmica/BarbaCKLRRV19}, we know that any explicit dynamic algorithm for $f(\alpha)$ colouring of a graph of arboricity $\alpha$, for any computable function $f$, needs $n^{\Omega(1)}$ time to process a single update. Hence, the only way of obtaining $\poly(\log(n))$ time algorithms for a colouring that uses a number of colours that depends only on $\alpha$ is to consider implicit dynamic algorithms, which is what we do in this paper. 
We call an algorithm that maintains the output explicitly an \emph{explicit dynamic algorithm}, and an algorithm that provides the access to the output via queries an \emph{implicit dynamic algorithm}. For most problems, one can easily derive implicit algorithms from explicit algorithms, and designing implicit algorithms is usually slightly simpler. 

When considering randomised dynamic algorithms, it is important to distinguish between the scenario in which the update and query sequences can be altered depending on the answer to earlier queries, and the scenario where the answer to previous queries does not change the next query asked. In the first setting, we say that the algorithm works against an \emph{adaptive adversary}, where as in the second scenario, we say that the adversary is \emph{oblivious}. Clearly, any deterministic algorithm also holds against an adaptive adversary.

\paragraph{The Simulation Framework}\label{simulation-framework}
The starting point of this paper is the observation that in order to determine an answer to a query in the implicit dynamic setting, one can simulate a \textbf{L}ocal \textbf{C}omputation \textbf{A}lgorithm \cite{DBLP:conf/innovations/RubinfeldTVX11}. While the complexity of the query-response of a dynamic algorithm might be larger than its LCA counterpart (as the analysis of an LCA algorithm does not include computation complexity), quite often the LCA algorithms are computationally simple, and the query complexity of an LCA algorithm is the deciding factor.
One way of obtaining LCA algorithms is to simulate distributed algorithms~\cite{DBLP:journals/tcs/ParnasR07}. The query resulting from a simulation of an $R$-round algorithm is just the size of the largest $R$ hop neighbourhood i.e.\ we obtain an LCA algorithm with $\bigO(\Delta^R)$ probe (query) complexity, where $\Delta$ is the maximum degree of the input graph. However, as observed in~\cite{DBLP:conf/wdag/GhaffariL17}, one can go one step further and design the distributed algorithm in a way, that makes the simulation less expensive. That is, one can design algorithms that make the outcome for a single vertex depend on some small subset of its $R$-hop neighbourhood, rather than whole $R$-hop neighbourhood. We will refer to size of this small subset as the \emph{volume} of the distributed algorithm. The volume of distributed algorithms has been studied before in~\cite{DBLP:conf/podc/RosenbaumS20}. Provided that a LOCAL or LCA algorithm has \emph{both} low volume and low time-complexity, we can use this simulation technique to obtain dynamic algorithms.

While we are not aware of any small volume algorithms that could be directly simulated, some of the papers on the subject present techniques that are useful here. The first relevant technique is a deterministic colour reduction technique~\cite{DBLP:journals/siamcomp/Linial92}, that can be also adapted to low arboricity graphs~\cite{DBLP:journals/dc/BarenboimE10}. The second relevant technique is a randomised colour assignment - while it does not guarantee that we colour all the vertices, it allows us to obtain some partial colouring that is easily extendable to a proper full colouring (see e.g. \cite{DBLP:conf/wdag/GhaffariL17}). Nevertheless, arboricity dependent colouring is a problem that in a distributed setting requires $\Omega(\log n)$ rounds, and that seems to have high volume (i.e., to our knowledge there are no results on the volume of arboricity dependent colouring, and there are also no reasons to believe it is significantly smaller than $\Delta^{\bigO(\log n)}$). Hence, in order to follow a string of reductions somewhat similar to the one described above, we need to make some modifications. 

One key observation is that some information is possible to maintain dynamically efficiently, even though that same information would  
require a large volume to compute in a distributed setting. In particular, we can maintain a \emph{bounded out-orientation} i.e.\ an orientation of the edges of a graph such that each vertex has bounded out-degree in $\bigO(\poly(\alpha, \log{n}))$ time dynamically. Then, we may consider a distributed model of computation that has access to this orientation, which allows us to design algorithms with drastically lower volume. Combined with other observations of a similar flavour and efficient sequential implementations of some subroutines, we arrive at two efficient implicit dynamic colouring algorithms matching the colouring guarantees of the best distributed algorithms for arboricity dependent colouring. 
Finally, we exploit additional symmetry breaking derived from the query sequence to provide an efficient implicit dynamic colouring algorithm that uses even fewer colours. We believe that some of the techniques developed to do so might be of independent interest to researchers studying dynamic as well as distributed models of computation.

\vspace{-5pt}
\subsection{Our Results} 
In this paper we propose three kinds of implicit algorithms for arboricity dependent colouring:
\begin{itemize}[leftmargin=5pt]
\item deterministic $\bigO(\alpha ^2)$ colouring algorithms with:
\begin{itemize}
	\item amortised $\bigO(\poly(\log n))$ update and query time, or 
	\item worst-case $\bigO(\poly(\alpha, \log n))$ update and query time, 
\end{itemize} 
\item a randomised algorithm providing $\bigO(\alpha \log \alpha)$ colouring w.h.p. in 
\begin{itemize}
	\item $\bigO(\poly \log n)$ worst-case update and query time, against an oblivious adversary, or
	\item $\bigO(\poly(\alpha, \log n ))$ worst-case update and query time, against an adaptive adversary	
\end{itemize}
\item a randomised algorithm providing $\betterColouringPaletteSize$ colouring w.h.p. in $\bigO(\poly \log n)$ amortised update and query time, against an oblivious adversary.
\end{itemize}

Our randomised algorithms hold under the oblivious adversary assumption, and their stated bounds on the running times hold with high probability.

\paragraph{Strengths}
All of those results provide implicit colourings with the number of colours polynomial in $\alpha$, which greatly improves over the current best algorithms for implicit colouring%
~\cite{DBLP:journals/corr/abs-2002-10142,DBLP:conf/icalp/CR22} that need a number of colours that is exponential in $\alpha$. Furthermore, the query- and update time complexities of our randomised algorithms do not depend on $\alpha$, and are $\polylog n$, even if $\alpha$ is large.

Thus, on one hand, in the small-arboricity realm, we provide an exponential improvement over state of the art. On the other hand, for arboricity $\alpha = \Omega(\log ^{\varepsilon} n)$, we reduce the problem to subproblems of arboricity $O(\log\log n)$, and this type of reduction does not translate directly into a distributed algorithm -- in spite of that it has distributed algorithms as its main inspiration. 
Curiously, even if, hypothetically, the distributed graph would come equipped with an asymptotically optimal out-orientation, our techniques would not directly translate.

\subsection{Overview of Techniques}
We will mainly design implicit dynamic colouring algorithms with a polynomial (in $\alpha$) amount of colours and an update and query time in $\bigO(\poly(\alpha, \log n))$. This is motivated by two things: first of all the explicit lower bound from~\cite{DBLP:journals/algorithmica/BarbaCKLRRV19} seem to hint at the fact that it is the most difficult to design such algorithms for $\alpha = \bigO(\log n)$ which in itself is a rather large class of graphs containing for example planar graphs. Second of all, and perhaps more importantly, because we will give a (rather standard) randomised reduction from the general case to the case where $\alpha = \bigO(\log n)$. The idea is that a random partition of the vertices will form subgraphs of low arboricity with high probability. We may then apply the algorithms to these low-arboricity subgraphs. We describe this reduction in more detail in~\cref{a:to_low_arboricity}. 
In light of this reduction, we lose nothing from limiting ourselves to the regime $\alpha = \bigO(\log n)$. Hence, we will focus on these low arboricity graphs for the rest of the paper, and then derive corollaries using the randomised reductions. 

As already hinted to in the introduction the starting point is the observation that we may simulate distributed algorithms to answer colouring queries, provided that the distributed algorithms have both low volume and low time-complexity. The first, fundamental observation is that in the dynamic setting, we can maintain some additional information that is difficult to compute in a distributed way -- i.e.\ a form of dynamic advice. In particular, the distributed arboricity dependent colouring indeed requires $\Omega(\log n)$ rounds~\cite{DBLP:journals/algorithmica/BarbaCKLRRV19}, but given a $d$ out-degree orientation, one can for instance compute an $\bigO(d^2)$ vertex colouring in $\bigO(\log^* n)$ rounds \cite{DBLP:journals/dc/BarenboimE10}. Hence, since the low out-degree orientation can be maintained in polylogarithmic time, for the purpose of simulation, we may assume that the distributed algorithm starts with knowledge of such an orientation. This already reduces the volume of the algorithm from $\Delta^{\Omega(\log n)}$ to $\Delta ^{\bigO(\log^* n)}$.
This observation is a cornerstone in all of our algorithms. 

For our deterministic algorithm, we also give the distributed algorithm implicit access to a colouring with $2^{\bigO(\alpha)}$ colours. This allows us to adapt the algorithm by \cite{DBLP:journals/dc/BarenboimE10} and further reduce this round complexity to be constant, independent of $\alpha$, resulting in low volume. The final step is then to show that we can implement all of the computational steps efficiently. To this end, we have to rely on explicit constructions in lieu of randomised existential proofs, which are normally sufficient for distributed algorithms. Luckily~\cite{erdos} provides an explicit construction that we can apply efficiently in a dynamic setting.

For the first randomised algorithm, we adapt the algorithm from \cite{DBLP:conf/wdag/GhaffariL17}, based on something similar to the graph shattering technique \cite{DBLP:conf/focs/BarenboimEPS12} which reduces the problem to many smaller subproblems. In the context of colourings, shattering typically works by running a simple colour-proposition experiment that leaves some vertices coloured and some vertices uncoloured. The idea is then that the uncoloured vertices form smaller subgraphs that corresponds to the smaller subproblems.
In the distributed setting, it is enough that these smaller subproblems have low diameter, and hence allow for very efficient distributed algorithms. In case of dynamic graph algorithms, in order to solve problem on the smaller subproblem in some trivial way we require not only that they  have low diameter, but also small volume. Hence, we have to consider directed shattering, and show that we get not only low diameter, but also low volume. Luckily, for our first randomised algorithm a fairly straight-forward modification of the argument from distributed setting also applies to shattering in our dynamic setting.  As for the small shattered components left behind - we can use a simple sequential algorithm to handle them. 

In the third and final algorithm, we go beyond this simulation with dynamic advice framework, and design a dynamic algorithm using fewer colours than what is obtainable by simulating known distributed algorithms (even if assuming advice in the form of edge-orientations). We arrive at this algorithm by addressing some of the limitations of the simple randomised algorithm described above. Again, we want to use the shattering technique~\cite{DBLP:conf/focs/BarenboimEPS12}, and therefore we briefly describe the general idea behind it and some of the problems that arise when trying to go below $\bigO(\alpha \log \alpha)$ colours, when $\alpha = o(\poly n)$. 
In the original setting, the shattering technique tries to assign some labels to vertices at random; a label can be either good (a feasible output / something that allows the algorithm to progress) or bad. The shattering technique essentially relies on the observation that if there is a bounded probability of a particular label being bad, and the random outcomes of the vertices are loosely dependent (or independent), then the connected components induced by the vertices with bad labels are small, and can be handled in some naive way.

For our purpose, this technique can't be applied directly, as we either have too high dependency between the outcomes, or too high probability of a label being bad. As such, we propose another approach, based on a random partition into graphs with low arboricity i.e.\ we want to construct an arb-defective colouring, as introduced and used  in~\cite{FraigniaudHK16,Barenboim16,BarenboimEG22,MausT20,Maus21}.  Given such an arb-defective colouring, we can further reduce the arboricity of the graphs we are trying to colour to the case where $\alpha = o(\poly \log n)$. It is however not clear how to construct such an arb-defective colouring. The reason is that the following trade-off should be satisfied:
suppose we wish to construct an arb-defective colouring that allows $g(\alpha)$ colouring conflicts while picking colours uniformly at random from a palette of size $h(\alpha)$. The colouring strategy has to satisfy two properties:
\begin{enumerate}
    \item The probability that a specific vertex fails in receiving a colour should be at most $1/\alpha^2$ in order for standard shattering arguments to be directly applied. 
    \item At the same time $h(\alpha)g(\alpha) = o(\alpha \log \alpha)$ in order for the defective colouring to be useful. 
\end{enumerate}
This combination leads to the following problem: On one hand, 
in order to satisfy 1), we have to use palettes large enough so that the probability that $g(\alpha)$ neighbours pick the same colour becomes smaller than $1/\alpha$ i.e.\ in order to directly apply a Chernoff bound, it seems difficult to avoid either $g(\alpha) = \Omega(\log \alpha)$, or $h(\alpha)$ to be so large that the probability that a vertex avoids $\Omega(\alpha)$ colours is smaller than $1/\alpha$. Again, in order to apply Chernoff bounds in a standard way, we need $h(\alpha) = \Omega(\alpha \log \alpha)$ which contradicts $2)$. On the other hand, consider the following graph: 
\begin{center}
\includegraphics[width=8cm]{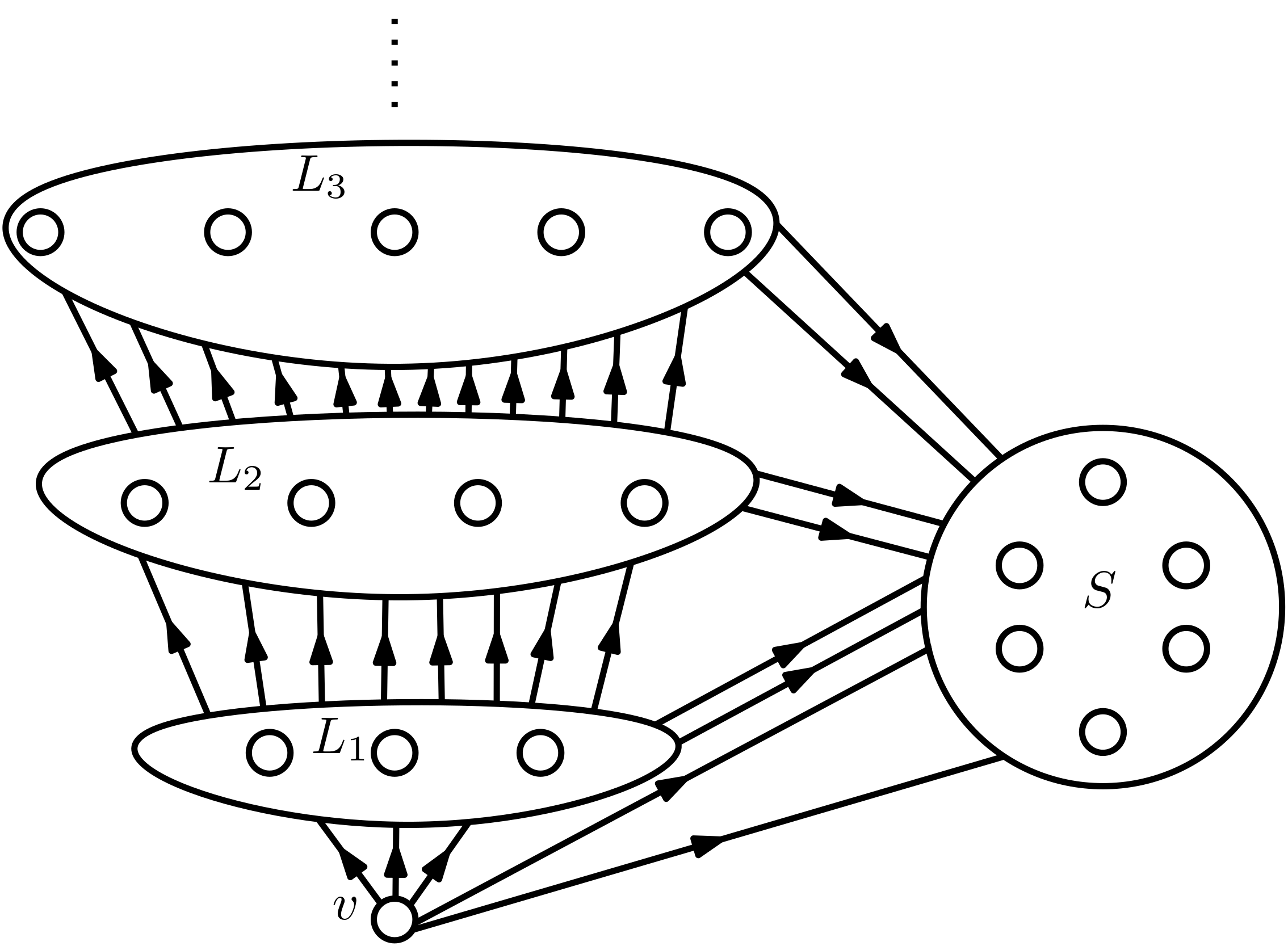}\\
\end{center}
The probability that all of the vertices in the set $S$ of size $g(\alpha)$ receives the same colours depends only on $\alpha$. So if $\alpha = O(1)$, and we consider many of such copies, then w.h.p. one will have $S$ be a monochromatic set. 
Now the probability that each vertex in $L_{i}$ chooses the same colour will be at least $1/h(\alpha)$. In particular, the expected size of $L_{i+1}$ would be roughly $\Omega(\alpha)/h(\alpha) = g(\alpha)/o(\log \alpha)$ times the size of $L_{i}$ i.e.\ the volume in front of $v$ would grow exponentially, and we would not get any shattering.\\
In Section~\ref{s:better_colouring}, we address this problem in steps. First we observe that we can get rid of some of these dependencies by using the query ordering to break symmetry. An amortised analysis then shows that we can propose colours so as to get very bounded dependence. However, this comes at the cost that now outcomes of colour experiments become correlated with the dependence between other experiments. Hence, we cannot apply standard shattering directly, as these techniques usually assume that all dependencies are given up front. 
Instead we can consider random subgraphs and show that each one has little probability of being part of a shattered component. We elaborate on this approach in Section~\ref{s:better_colouring}.

\subsection{Further Related Work}
\paragraph{Dynamic Graph Algorithms}
Graph colouring being NP-hard in general, we only have hope for polylogarithmic update-time dynamic algorithms in restricted cases for which a (near-) linear time algorithm is known. Such restricted graph classes include graphs of bounded max degree, interval graphs, planar graphs, minor-free graph classes, forests, and graphs of bounded arboricity.

While forests are straightforward to vertex-colour in the static setting, they surprisingly do not admit explicit dynamic recolouring algorithms~\cite{DBLP:journals/algorithmica/BarbaCKLRRV19}. As the authors of \cite{DBLP:journals/corr/abs-2002-10142} also point out, \cite{DBLP:journals/algorithmica/BarbaCKLRRV19} shows that one cannot hope for an algorithm that maintains an $f(\alpha)$ colouring explicitly, for any function $f$ parameterised only by the arboricity; there must be some dependency on $n$. This is the main reason to study the implicit colouring problem, in which we need to report a colour only as an answer to a query. 
In \cite{DBLP:journals/corr/abs-2002-10142} the authors provide an implicit dynamic $2^{O(\alpha)}$-colouring algorithm with polylogarithmic update time, surpassing the lower bound for explicit algorithms.
For explicit colouring, Solomon and Wein give an algorithm for $\tilde{O}(C/\beta \cdot \log ^2 n)$ colouring with $O(\beta)$ recourse, and for $O(\alpha\log^2 n)$ colouring in amortised constant time.


For graphs of bounded maximum degree, there has been rapid improvement during the past years~\cite{DBLP:conf/soda/BhattacharyaCHN18,DBLP:conf/stacs/Henzinger020,DBLP:journals/talg/HenzingerP22,DBLP:journals/talg/BhattacharyaGKL22}, and currently there exist randomised dynamic algorithms with an update time that is constant with high probability~\cite{DBLP:journals/talg/BhattacharyaGKL22}. 
Another class of easily colourable graphs is \emph{interval graphs}, for which colouring is highly motivated due to the applications in scheduling. Here, already decades ago it was shown that efficient dynamic updates to a dynamic $k$-colourable interval graph are impossible to obtain~\cite{KiersteadTrotter1981extremal,DBLP:journals/ita/ChrobakS88}. Recently, the topic of dynamic interval graphs has been studied both in the approximate~\cite{DBLP:conf/swat/BosekDFPZ20} setting and the setting restricted to unit intervals~\cite{BosekZ22}.

For the dynamic \emph{bounded out-degree orientation} problem, Brodal and Fagerberg achieve an $O(\alpha)$ out-orientation in $O(\alpha + \log n)$ (amortised) update time~\cite{Brodal99dynamicrepresentations}. 
Henzinger, Neumann, and Wiese~\cite{DBLP:journals/corr/abs-2002-10142} give a different $O(\alpha)$ out-orienting algorithm with $O(\log ^2 n)$ amortised update time. 
For worst-case update time arboricity decompositions, the current best is by Christiansen et al~\cite{AOOA-CHHRS} who decompose into $O(\alpha)$ forests having $O(\log^2 n \log \alpha)$ worst-case update time. 
Other trade-offs include: $O(\alpha \log n)$ out-edges in constant amortised update time~\cite{10.1016/j.ipl.2006.12.006}, $O(\alpha + \log n)$ out-edges in $O(\log n)$ worst-case update time~\cite{berglinetal:LIPIcs:2017:8263}, and $(1+\varepsilon)\alpha +2$ out-edges in $O(\log^3 (n)\alpha^2/\varepsilon^6)$ worst-case update time~\cite{DBLP:conf/icalp/CR22}. See also~\cite{He2014OrientingDG,kopelowitz2013orienting,benderESA21}.


\section{Preliminaries} \label{sec:prelim}
Here, we state well-known results on maintaining out-orientations on dynamic graphs. The first result states that we can maintain an acyclic out-orientation with low out-degree:

\begin{theorem}[\cite{DBLP:conf/icalp/CR22}]\label{thm:acyclic}
 Let $G$ be a dynamic graph on $n$ vertices undergoing insertions and deletions such that the arboricity of $G$ never exceeds some maximum 
$\alpha_{max}$. Then there exists an algorithm maintaining an acyclic $2(\alpha_{max}+1)$ out-orientation of $G$ with an amortised update time of $\bigO{(\alpha_{max}^2 \log n)}$.
\end{theorem}
Note that even if we do not know an upper bound the arboricity $\alpha_{\max}$, as explained in~\cite{DBLP:conf/icalp/CR22} we can use the ideas of~\cite{DBLP:conf/stoc/SawlaniW20} in order to maintain $\bigO{(\log n)}$
 copies of the graph -- each with different upper bounds on the arboricity -- as well as a pointer to a copy with an upper bound close to the actual arboricity of the graph in poly-logarithmic update time. 

The above data structure~\cite{DBLP:conf/icalp/CR22} maintains a list of out-neighbours for each vertex, and supports access to the $k$ out-neighbours in $O(k)$ time.

However, if one is only interested in constant factor approximations of the arboricity, more recent work gives improved running times.

\begin{theorem}[\cite{AOOA-CHHRS}] \label{thm:approximatealpha}
	Let $G$ be a dynamic graph on $n$ vertices undergoing edge insertions and deletions. 
	Then there exists an algorithm that:
	\begin{itemize}
		\item explicitly maintains an arboricity decomposition into $O(\alpha)$ forests,
		\item explicitly maintains a $(1+\epsilon)$-approximation to the fractional arboricity, and thus, a constant-factor approximation to the arboricity
		\item explicitly maintains an out-orientation with $O(\alpha)$ out-edges,
	\end{itemize}  
Its worst-case update time is $O(\log^2 n \log \alpha)$, where $\alpha$ is the arboricity at the update time.
%
\end{theorem}
\begin{corollary}[\cite{AOOA-CHHRS,DBLP:journals/corr/abs-2002-10142}]
\label{thm:exp_colouring_wc}
Let $G$ be a dynamic graph on $n$ vertices undergoing insertions and deletions. Then there exists an algorithm that maintains an implicit $2^{\bigO{(\alpha)}}$ colouring with worst-case update time in $\bigO{(\log^3 n \log \alpha)}$ and query time $\bigO{(\alpha \log{n})}$. Here, $\alpha$ is the current arboricity of the graph. 
	
\end{corollary}

%
Finally, the following theorem shows how to efficiently maintain an implicit vertex colouring:
\begin{theorem}[\cite{DBLP:journals/corr/abs-2002-10142}] \label{thm:exp_colouring}
 Let $G$ be a dynamic graph on $n$ vertices undergoing insertions and deletions, and let $\alpha$ be the arboricity of the graph. Then there exists an algorithm that maintains an implicit $2^{\bigO{(\alpha)}}$ colouring with amortised update time in $\bigO{(\log^3 n)}$ and query time $\bigO{(\alpha \log{n})}$.
\end{theorem}

\subsection{Tools for Randomised Algorithms} 
As mentioned earlier, in \cref{s:simple_randomized,s:better_colouring} we mostly focus on graphs with arboricity $\bigO(\log n)$. The reason for that is that we can use a dynamic arboricity approximation 
together with random partition of the vertices to reduce the general case to the case with arboricity $\bigO(\log n)$. Overall, the reduction only increases the number of colours by a constant factor. 
\begin{lemma}\label{lem:reduction}
Assume that there exists an algorithm $A$ that maintains a $f(\alpha_{\max})$ colouring of the graph in $U_A(\alpha_{\max}, \log n)$ update and $Q_A(\alpha_{\max}, \log n)$ query time, where $\alpha_{\max}$ is a known upper bound on the arboricity of the dynamic graph. Then, there exists a randomised algorithm $A'$ that maintains a colouring with $(1+\eps)\lceil \frac{\alpha}{\log n} 2f(\log n) \rceil$ colours, in $U_A(\alpha, \log n)$ amortised update time w.h.p.\, and $Q_A(\bigO(\log n), \log n)$ query time, where $\alpha$ is the current arboricity of the graph.
\end{lemma}

This reduction is a fairly simple application of known techniques, hence we provide a proof of \cref{lem:reduction} only in \cref{a:to_low_arboricity}.

Furthermore, both our randomised algorithms rely on the graph shattering technique, which was first used in the context of distributed algorithms in \cite{DBLP:conf/focs/BarenboimEPS12}. Roughly speaking, algorithms based on this technique try to find some partial and extendable solution to a particular problem in a way that guarantees that the graph induced by vertices that are not covered by the partial solution consists of low diameter connected components of small volume. We present a more formal statement in \cref{lem:shattering}.

\begin{lemma}\label{lem:shattering}\cite{DBLP:conf/focs/BarenboimEPS12}
Consider an $n$-vertex graph $G$ with maximum degree $\Delta$. Let $S$ be a set of randomly chosen vertices, to which we include each vertex independently, with probability $p \leq \frac{1}{\Delta^4}$. Let $R(v)$ be the set of vertices reachable from $v$ in the subgraph of $G$ graph induced by $S$ in the oriented graph resulting from the bounded out-orientation algorithm. Then, $\card{R(v)} \leq (c+1)\log_{\Delta} n$, for all $v \in G$,  with probability at least $1-n^{-c}$.
\end{lemma}
\begin{proof}
We can apply arguments similar to those presented in~\cite{DBLP:conf/focs/BarenboimEPS12}:
suppose $|R(v)| \geq (c+1)\log_{\Delta}(n) = r$ for some $v\in G$. Then $R(v)$ contains a sub-tree $T$ of size $r$ such that for every vertex $u \in T$ there exists a path from $v$ to $u$ contained in $T$.
The number of rooted, unlabeled trees of this size is at most $2^{2r}$ since each tree is characterised by its Euler-tour which can be encoded using $2r$ bits. 
In particular, there are at most $2^{2r}\Delta^{r-1} \leq n^{3(c+1)}$ ways to both choose and embed $T$ into the graph with $v$ as the root. 
By assumption, the probability that one such choice of $T$ is uncoloured is at most $\Delta^{-4r} = n^{-4(c+1)}$.
Hence a union bound over all choices of $v$ and $T$ shows that the probability that $S$ induces one such $T$, or equivalently that $|R(v)| \geq (c+1)\log_{\Delta}(n) = r$ for some $v\in G$, is no more than $n\cdot{}n^{3(c+1)} \cdot{} n^{-4(c+1)} \leq n^{-c}$.
In particular, $|R(v)| \leq r$ for all $v \in G$ with probability at least $1-n^{-c}$.
\end{proof}

While the technique was introduced in \cite{DBLP:conf/focs/BarenboimEPS12}, the authors never state it in a general way, and they simply give several more tailored proofs for particular problems, following the same approach. For the sake of completeness, we include the proof of \cref{lem:shattering} following the same approach.

One can rarely use this general statement as it is, in a black box manner. The reason is that in most cases, the events of including vertices in $S$ were not independent. Still, sometimes the case is that the outcome for any particular $v$ depends on some random decisions made only in its $k$-hop neighbourhood. Then, any two outcomes for vertices at distance at least $2k+1$ are independent. Hence, one could build a graph $G^{2k+1}$ that connects vertices at distance at most $2k+1$ and use the general shattering statement. Then, the volume in the original graph is at most $\Delta^{2k+1}$ times larger than a connected component in $G^{2k+1}$.


\subsection{Further notation}

For any vertex $v$ in an oriented graph, we let $N^+(v)$ denote its in-neighbourhood, and symmetrically $N^-(v)$ its out-neighbourhood. Here, an oriented graph is an undirected graph with an orientation enforced upon its edges, not to be confused with a directed graph (i.e. the input graph is undirected, and the orientation of edges might be adjusted by the algorithm after processing some updates). 
Given a subset $V'\subset V$ of the vertices, we use $G[V']$ to denote the the subgraph induced by the vertices of $V'$. 
For any subgraph $H$ of $G$, let $d_H(v)$ denote the degree of $v$ in the subgraph $H$. Similarly, $d^+_H(v)$ and $d^-_H(v)$ are its in- and out-degree in the subgraph $H$.

Recall that a graph has degeneracy $d$ if every subgraph has a vertex of degree $\le d$, i.e. it possible to recursively delete (or, `peel off',) a vertex of degree $\le d$ until the empty vertex set remains. Such a recursive deletion order is called a \emph{degeneracy order}, and its reverse immediately leads to a $d+1$ colouring which we call a \emph{degeneracy colouring}. Note that neither the degeneracy ordering nor the degeneracy colouring is necessarily unique.




\section{Deterministic $\bigO(\alpha^2)$ colouring} \label{s:det_colouring}
In this section, we present our deterministic algorithm for $\bigO(\alpha^2)$ vertex colouring. 

\begin{theorem}\label{thm:det_colouring}
There exists a deterministic dynamic algorithm that for graph of arboricity $\alpha$ maintains an implicit $\bigO(\alpha^2)$ vertex colouring, with $\bigO(\log \alpha \log^3 n)$ worst-case update time and $\bigO(\alpha^5 + \alpha^3 \log n)$ query time.
\end{theorem}

As previously discussed the idea is to answer a colour query by simulating a distributed algorithm -- more specifically the one presented in~\cite{DBLP:journals/dc/BarenboimE10}. We show this theorem in the following section, but first let us take a look at its implications in \cref{cor:polylog_implicit}, which says that we can `switch' to an explicit algorithm with a different trade-off when $\alpha$ exceeds $\log n$ in order to get both an update time and a query time independent of $\alpha$.

\begin{corollary}\label{cor:polylog_implicit}
There exists a deterministic algorithm that maintains an adaptive implicit $\bigO(\alpha^2)$ colouring, with $\bigO(\log^3 n)$ amortised update time, and $\bigO(\log^5 n)$ query time.
\end{corollary}
\begin{proof}
From \cite[Theorem 4]{DBLP:journals/corr/abs-2002-10142} one can maintain an implicit $\bigO(\alpha \log n)$ colouring with $\bigO(\log^2 n)$ amortised update time, and $\bigO(\log n)$ query time. Furthermore, \cite{DBLP:journals/corr/abs-2002-10142} also shows how to maintain a constant approximation of the current arboricity in $\bigO(\log^2 n)$ amortised time. Hence we can maintain the data structures for both algorithms, the one from \cite{DBLP:journals/corr/abs-2002-10142}, and the one presented below, which takes amortised $\bigO(\log^3 n)$ time per update if we apply the out-orientation from Theorem~\ref{thm:exp_colouring}. 

Whenever $\alpha \in \bigO(\log n)$, we query our algorithm, and get a colour from a palette of size $\bigO(\alpha^2)$. Otherwise, we query the algorithm from \cite{DBLP:journals/corr/abs-2002-10142}, and obtain a colour from palette of size $\bigO(\alpha \log n)$, which for $\alpha \in \Omega(\log n)$ is also $\bigO(\alpha ^2)$. In the worst case, it takes $\bigO(\log^5 n)$ time, as the worst case complexity occurs when we get a query to our algorithm, when $\alpha \in \Theta(\log n)$.
\end{proof}

\paragraph{Simulating the distributed algorithm efficiently} As briefly discussed in \cref{simulation-framework} p.\pageref{simulation-framework}, we cannot rely on a straightforward simulation in order to get our implicit colouring algorithm. Here we briefly discuss how we address the bottlenecks to obtain an algorithm with update and query complexity $\poly(\alpha, \log n)$, before we describe the algorithm in more details.

The first, fundamental observation we have seen before: namely that in the dynamic setting, we can maintain some additional information or advice that is difficult to compute in a distributed way. In particular, the arboricity dependent colouring indeed requires $\Omega(\log n)$ rounds~\cite{DBLP:journals/algorithmica/BarbaCKLRRV19}, but given a $d$ out-degree orientation, one can compute $\bigO(d^2)$ vertex colouring in $\bigO(\log^* n)$ rounds \cite{DBLP:journals/dc/BarenboimE10}. Hence, since the low out-degree orientation can be maintained in polylogarithmic time, for the purpose of simulation we may assume that the distributed algorithm is given such an orientation as advice. This already reduces the volume of the algorithm from $\Delta^{\Omega(\log n)}$ to $\Delta ^{\bigO(\log^* n)}$.

The next observation that further reduces the volume is that given the $d$ out-degree orientation of the edges, the output of a vertex $v$ actually depends only on the topology of the graph induced by the vertices reachable from $v$ via paths of $\bigO(\log^* n)$ directed edges. This observation reduces the volume further down, to $d ^{\bigO(\log^* n)}$.

The final observation uses the fact that if we are given the $d$ out degree orientation as advice, and vertices have initial colours from the range $[1, \dots, \poly(c)]$, the algorithm from ~\cite{DBLP:journals/dc/BarenboimE10} needs only $\bigO(\log^* c)$ rounds to find an $\bigO(d^2)$ colouring. In particular, this algorithm consists of two (almost identical) $\bigO(1)$ round building blocks, the first reducing the size of the palette from $\poly(c)$ to $\bigO(d^2\log c)$, the second - from $\bigO(d^3)$ to $\bigO(d^2)$. 

By \cref{thm:exp_colouring} we know that we can efficiently maintain a $2^{\bigO(\alpha)}$ colouring, hence we can assume that the distributed algorithm also receives such a colouring as advice. That is, the distributed algorithm we want to simulate gets both an $\bigO(\alpha)$ out-degree orientation and an $2^{\bigO(\alpha)}$ colouring as advice. We use one round of the distributed algorithm to reduce the palette size to $\bigO(\alpha^4)$, and then another round to reduce it further down to $\bigO(\alpha^2)$. Note that in our case, we reduce the number of colours from $\bigO(\alpha^4)$ to $\bigO(\alpha^2)$ rather than from $\bigO(\alpha^3)$ to $\bigO(\alpha^2)$ -- the technique is the same; later, we explain that using it one can obtain a one round reduction from $\poly(\alpha)$  to $\bigO(\alpha^2)$.

Overall, the volume needed to determine the output corresponds to all vertices reachable over directed paths of length 2, and so the volume is is $\bigO(\alpha^2)$, i.e.\ small. However, we still need to show that the time complexity is small, which requires us describe the the algorithm in more detail and explain how to implement it in a sequential fashion. We do this below.

\paragraph{The details of the algorithm}
First we recall the how the distributed algorithm works: the general idea behind the distributed algorithms in~\cite{DBLP:journals/dc/BarenboimE10, DBLP:journals/siamcomp/Linial92} is to use $r$-cover-free families of sets. Here, a family of sets $\mathcal{A}$ is said to be \emph{$r$-cover-free} if for all choices of $r+1$ sets $A_{1}, \dots, A_{r+1} \in \mathcal{A}$ ($A_{r+1} \neq A_i$ for $1\leq i \leq r$), it holds that $A_{r+1}$ is not contained in the union of the sets $A_{1}, \dots, A_{r}$.    
The idea is that if one can find an $r$-cover-free family of size at least $k$, where every set is a subset of a set of size $\bigO(\alpha^2 f(k))$ for some function $f$, then one can reduce the number of colours from $k$ to $\bigO(\alpha^2 f(k))$ by assigning each colour-class a unique set from $\mathcal{A}$, and then assigning every vertex a colour from this set not present in the set of any of the at most $d$ out-neighbours' set. This is always possible if $r \geq d$, since then, by construction, the neighbours' sets are certain to never cover the set of the vertex.
Continuing like this eventually yields a colourings with few colours. 
It was shown in~\cite{erdos,DBLP:journals/siamcomp/Linial92} that one can find such an $r$-cover-free family with $f(k) = \log k$. This construction is probabilistic, however, as pointed out in~\cite{DBLP:journals/siamcomp/Linial92} one can use a construction based on polynomials~\cite{erdos} to explicitly construct an $r$-cover-free family with $f(k) = \log^2 k$. As noted in both~\cite{DBLP:journals/dc/BarenboimE10,DBLP:journals/siamcomp/Linial92} this construction can be improved to $f(k) = c$ when $k = \bigO(\alpha^2 \log^2 \alpha)$. 
In order to explicitly describe and analyse the time complexity of our algorithm, we state these constructions explicitly below.
\begin{lemma}\label{lma:colred}
(\cite{erdos,DBLP:journals/siamcomp/Linial92}) For any $k \leq 2^{\bigO(\alpha)}$ and any $r \geq \alpha$, there exists an $r$-cover-free family $\mathcal{A}$ of size at least $k$ such that every set in $\mathcal{A}$ is a subset of a set of size $r^2 \cdot{} \log^2 k$. Furthermore, there exists an injection $g: [k] \rightarrow \mathcal{A}$, such that given $x \in [k]$ one can compute $g(x)$ deterministically in $\bigO(d)$ time, and given $r+1$ sets $A_1, \dots, A_{r+1}$, we can determine an element of $A_1$ not covered by all other sets in $\bigO((d \cdot{}r)^2)$ time.
\end{lemma}

\begin{proof}
In our algorithm, we use the construction based on polynomials, from ~\cite[Example 3.2]{erdos}, setting $d = \log k$ and $q = r\log k$. For the sake of completeness, we briefly recall this construction here.

The general idea is to use the fact that two polynomials of degree $d$ can intersect in at most $d$ points. Hence, one way of constructing an $r$-cover free family is to consider a family of polynomials of degree at most $d$ on $Z_q$, where $d,q$ are chosen such that the following inequalities hold:
\begin{itemize}
\item $r \cdot d  < q$, to guarantee that each polynomial $P$ is defined on enough points for the set $\set{(i,P(i)) \| i \in Z_q}$ to not be covered by the union of such set for $r$ other polynomials, no matter how they are chosen.
\item $q^{d+1} > k$, so that every colour from $[k]$ can be mapped via $g$ to precisely one polynomial.
\end{itemize}
These inequalities are sure to hold if we just pick $d =  \Omega(\alpha)$, and $q$ to be the smallest prime satisfying $r \cdot d  < q$. Note that this implies that $q = \bigO(rd)$.
Now the sets for our $r$-cover-free family are defined as the graphs of all polynomials on $Z_q$ with degree at most $d$. That is given polynomial $P$, the corresponding set is $\set{(i,P(i)) \| i \in Z_q}$. Note that this family is $r$-cover-free since any polynomial covers at most $d$ such points, and by construction $r\cdot d  < q$. Observe that the total number of colours after the reduction is $q^2$. 

The time complexity of evaluating $g$ is $\bigO(d)$, as we can consider all polynomials $d$ in lexicographical order (defined by the coefficients), and map the number $i$ to the $i$'th polynomial in the lexicographical ordering (which can be done in $\bigO(d)$ time).  

Given a polynomial $P$, and $r$ other polynomials $P_1, P_2, \dots, P_{r}$, we determine an uncovered element $x$ by naively evaluating all the polynomials on all arguments from $Z_q$. The time complexity of such a naive approach is $(r+1) \cdot q \cdot (d+1) \in \bigO((dr)^2)$.
\end{proof}
Note that if $k = \bigO(\alpha^4)$, then we may pick $q  = \Theta(\alpha)$ and $d = 4$ in order to reduce the number of colours to $ \bigO(\alpha^2)$.

The algorithm is now simple to describe: when a vertex is queried it has to first determine its colour after the first colour reduction step. To do so, it determines its colour using the implicit colouring algorithm, maps it to a polynomial, asks its out-neighbours to do the same, and finally it picks a free colour as described above. 
To do the next and final colour reduction step, it does the same, but it maps the reduced colour of itself and its neighbours to a new polynomial and picks a new colour as before.
See~\cref{alg:reccol} for pseudo-code. For our purpose, we only need two iterations, hence we invoke it with $j=2$.

\begin{algorithm}
$\texttt{reduce colour}(v,j)$: \\
\ if $j = 0$: \\
\ \ $\texttt{return}$ implicit colour of $v$. \\
\ else: \\
\ \ for $w \in N^{+}(v)$: \\
\ \ \ set $k_w = \texttt{reduce colour}(w,j-1)$ \\
\ \ \ set $A_w$ to be the $k_w$'th set from $\mathcal{A}_{j}$ \\
\ \ $\texttt{return}$ free colour from $A_{v} \mathbin{\big\backslash}
\left(\bigcup \limits_{w \in (N^{+}(v) - \{v\})} A_w\right)$. \\
\caption{Deterministic Colouring} \label{alg:reccol}
\end{algorithm}

\paragraph{Analysis} 
We begin by analysing the time-complexity of performing a query.
\begin{lemma} \label{lma:querytime}
The query-algorithm has a worst-case time complexity of $\bigO(\alpha^5+\alpha^3 \log n)$.
\end{lemma}
\begin{proof}
The colour reduction begins from a  $2^{\bigO(\alpha)}$ colouring, and the $\bigO(\alpha)$ out-degree orientation is known up front. We set $d = \log 2^{\bigO(\alpha)} = \bigO(\alpha)$ and  $r' \in \Theta(\alpha)$. By~\cref{lma:colred}, we can therefore, assuming we know the $2^{\bigO(\alpha)}$ colouring, perform the first colour reduction for a vertex in time $\bigO(\alpha^4)$. In order to do the second colour reduction, we need to perform such a reduction for the $\bigO(\alpha)$ out-neighbours of the queried vertex. Hence the total time spent performing the first colour reduction is $\bigO(\alpha^5)$. For the second colour reduction, as also noted above, we pick a sufficiently large $q \in \bigO(\alpha)$ and $d = 4$, which yields an $\bigO(\alpha^2)$ colouring. By~\cref{lma:colred} the time complexity of this reduction is only $\bigO(\alpha)$ per vertex, and we only need to perform it for the queried vertex. Thus the overall time complexity for performing these reductions are in $\bigO(\alpha^5)$.

Note, however, that above we assumed that we initially new the $2^{\bigO(\alpha)}$ colouring and the out-orientation. Accessing this advice from the data structures that maintain it is $\bigO(\alpha \log n)$ resp.\ $\bigO(\alpha)$ per vertex. Since we only need this information for $\bigO(\alpha^2)$ vertices, the necessary advice can be fetched in $\bigO(\alpha^3 \log n)$ time per query. 
\end{proof}
Now~\cref{thm:det_colouring} easily follows:
\begin{proof}[Proof of~\cref{thm:det_colouring}]
Correctness of the algorithm follows from~\cref{lma:colred}. The update-time corresponds to the update time of the used out-orientation algorithm. Picking for instance the one from \cite{AOOA-CHHRS}, i.e. \cref{thm:approximatealpha} implying \cref{thm:exp_colouring_wc}, we get the claimed update time. Finally, the query complexity follows from~\cref{lma:querytime}.
\end{proof}

\section{Randomised $\bigO(\alpha \log \alpha)$ colouring} \label{s:simple_randomized}
In this section, we present an algorithm that maintains an implicit $\bigO(\alpha \log \alpha)$ vertex colouring of a graph with arboricity $\le \alpha$. The algorithm is an adaptation of a distributed arboricity-dependent colouring by Ghaffari and Lymouri~\cite{DBLP:conf/wdag/GhaffariL17} to the dynamic graph setting.
The relevance of this algorithm to the better colouring algorithm is twofold. Firstly, it presents a way of thinking that allows us to transform sequential, distributed or parallel algorithms into data structure for dynamic graphs. Secondly, it allows us to find the bottlenecks in the approach inspired by the algorithm of Ghaffari and Lymouri, which then are addressed in \cref{s:better_colouring}.
More formally, we show the following:
\begin{theorem}\label{thm:simplerand}
There exists a randomised dynamic algorithm that for graph of arboricity $\alpha$ maintains an implicit $\bigO(\alpha\log \alpha)$ vertex colouring with $\bigO(\log(\alpha)\cdot{}\log^2(n))$ worst-case update time and $\bigO(\alpha^5 + \alpha^3 \log n)$ query time with high probability.
\end{theorem}

\paragraph{Description of the Algorithm}
Once again, we provide a $d \in \bigO(\alpha)$ out-degree orientation as dynamic advice. Given that orientation, the algorithm tries to shatter the graph by performing a simple random colour-picking experiment at every vertex. Finally, it colours the the components induced by the failed experiments in a sequential manner.
The experiment is as follows: a vertex picks $4\log d$ colours -- each colour is picked u.a.r. from disjoint palettes of size $8d\log d$. The experiment at a vertex $v$ is a success if it picks a colour $\kappa$ such that no out-neighbour of $v$ also picked $\kappa$. If this is the case, $v$ colours itself with $\kappa$ -- and returns this colour as a response to a query. Otherwise, the experiment fails and $v$ is uncoloured. Then $v$ determines all failed colour experiments reachable from $v$ through directed paths only going through vertices with failed experiments. In other words, $v$ determines everything it can reach in the graph $G_{\mathit{exc}}$ induced by all vertices with failed colour experiments.
Now, the algorithm colours this subgraph with a separate palette in a sequential fashion: since $G_{\mathit{exc}}$ has arboricity $\alpha$, it is $2\alpha$ degenerate. In particular, we may pick a degeneracy ordering of $G_{exc}$ and colour the vertices of the subgraph according to this ordering. Note that every vertex picks a free colour not present at any out-neighbour in the original graph or any vertex before it in the ordering. See the Algorithms~\ref{alg:simpleCol}, ~\ref{alg:simplePropCol} and~\ref{alg:colour_residual}  for pseudo-code outlining the approach.

\begin{algorithm}
ProposeColour$(v)$ \\
if $\kappa(v) = \neg$ \\
\ \ ColourResidual $(v)$
\caption{Colour(v)} \label{alg:simpleCol} \label{alg:simple}
\end{algorithm}

\begin{algorithm}
if $\kappa(v) = \neg$ \\
\ \ for all $u \in N^{+}(v) \cup \{v\}$ \\
\ \ if $u$ have not already proposed colours \\
\ \ $C_u \gets$ $4\log d$ colours chosen uniformly at random from $[8d \log d]$ \\
if $\exists c \in C_v \setminus \bigcup\limits_{u \in N^+(v)} C_u$ \\
\ \ $\kappa(v) \gets c$  \\
return $\kappa(v)$ 
\caption{ProposeColour(v)} \label{alg:simplePropCol}
\end{algorithm}

\begin{algorithm}[H]
\KwIn{
$v \leftarrow $ a vertex to be coloured \\
$P \leftarrow $ palette of colours of size $3d+1$ reserved for $G_{\mathit{exc}}$ \\
}
$R(v) \gets v$\\
$Q.$push$(v)$ \\
while(Q is not empty) \\
\ \ $u \gets Q.$pop() \\
\ \ for $w \in N^+(u)$ \\
\ \ \ \ $s(w) \gets Colour(w)$ \\
\ \ \ \  if $s(w) = \neq$ \\
\ \ \ \ \ \ $Q.$push$(w)$ \\
\ \ \ \ \ \ $R(v) \gets R(v) \cup \{w\} $ \\
Compute a degeneracy ordering $x_1, x_2, \dots, x_k$ of $G[R(v)]$ \\
for $i = k$ to $1$ \\
\ \ $\kappa(x_i) \gets$ arbitrary colour from $P \setminus \bigcup \limits_{y \in N^+(x_i)} \kappa(y) \setminus  \bigcup\limits_{j > i} \kappa(x_j)$
\caption{ColourResidual $(v)$}\label{alg:colour_residual}
\end{algorithm}

\paragraph{Analysis}
Let us first consider the palette size needed: clearly the colouring experiment uses $\bigO(d \log d)$ colours. By~\cref{lem:palette_for_residual} below, we can colour all of the shattered components using $\bigO(d)$ colours. This results in a total of $\bigO(d \log d)$ colours being used. Furthermore, by~\cref{lem:palette_for_residual} the colouring is also proper. 
\begin{lemma}\label{lem:palette_for_residual}
Any subgraph $G_{\mathit{exc}}$, constructed as above, of a shattered component can be coloured with the same palette $S$ of $3d+1$ colours without introducing colour collisions in $G$. Furthermore, the time spend colouring such a subgraph is proportional to its size times $d$.
\end{lemma}
\begin{proof}
In order to prove \cref{lem:palette_for_residual}, we rely on two properties. First, there are no colours in $G$ that are coloured with a colour from $S$ and point to $G_{\mathit{exc}}$. The reason is that we colour everything in front of a vertex reachable via directed paths through vertices where the experiment failed. Hence we reach a contradiction, if some vertex that received a colour from $S$ was pointing to an uncoloured vertex.
  
The second property is that we are given a $d$ out-degree orientation of $G_{\mathit{exc}}$, which also means that the graph has degeneracy at most $2d$ and hence it contains a $2d$ acyclic out-orientation that can be computed e.g. by a greedy algorithm. Finally, we may colour each vertex according to a degeneracy order imposed by the acyclic out-orientation. When a vertex is coloured it has at most $2d$ neighbours in front of it in the degeneracy ordering (or topological ordering) and at most $d$ out-neighbours from the original out-orientation. Hence, it may pick a colour that neither of these vertices picked, and so we get no colour collisions. 
 
Finally, we may compute both the acyclic out-orientation as well as the greedy colouring in time proportional to its size times $d$, as we may have to consider out-neighbours not included in $G_{\mathit{exc}}$.
\end{proof}
In order to analyse the query complexity, we note that it is straight-forward to analyse the time spent performing the colour experiments. It is also clear that we may determine and
eventually colour $G_{\mathit{exc}}$ in time $\bigO{(|G_{\mathit{exc}}|}\cdot{}d)$. It is however not immediately clear how large $G_{\mathit{exc}}$ can be. 
We can however upper bound this size by using combining arguments from~\cite{DBLP:conf/wdag/GhaffariL17} with~\cref{lem:shattering} :
\begin{lemma}\label{lma:shattersizealg2}
Let $G_{\mathit{exc}}$ be as above i.e.\ it is the subgraph induced by all vertices with failed colour experiments over directed path beginning from some vertex $v$. Then with probability $1-n^{-c}$, we have $|G_{\mathit{exc}}| \leq (c+1)\log_{d} n$.
\end{lemma}
\begin{proof}
The first relevant property is that no matter the random choices of other vertices, a fixed vertex $v$ ends up in  $G_{\mathit{exc}}$ with probability at most $\frac{1}{d^4}$. This follows from the observation that no matter what the choices of all out-neighbours half of the colours remain free. Hence, any colour picked uniformly at random does not collide with the colours proposed by the out-neighbours w.p. at least $\frac{1}{2}$. Given that a fixed vertex chooses $4 \log d$ colours, the probability that all of them collide is at most $\paren{\frac{1}{2}}^{4 \log d} \leq 1/d^4$.

Let us consider random variables indicating whether particular vertices are included in $G_{\mathit{exc}}$. By the discussion above, they can be stochastically dominated by independent random variables that are equal to one with probability $1/d^4$. Hence, we may now apply~\cref{lem:shattering} to see that $G_{\mathit{exc}}$ has size $\bigO(\log_{d} n)$ with high probability. 
\end{proof}
We can now combine the above arguments to analyse the time complexity of a query. 
\begin{lemma}\label{lma:querycomprand}
The algorithm answers a query in $\bigO(d^2 \log n)$ time with high probability.
\end{lemma}
\begin{proof}
It takes $\bigO(d \log d)$-time to perform the colour experiment for each vertex. By~\cref{lma:shattersizealg2}, we perform the experiment at $\bigO(d \log_d n) = \bigO(\frac{d \log n}{\log d})$ vertices with high probability. 
By~\cref{lma:shattersizealg2,lem:palette_for_residual} we can colour the residual graph $G_{\mathit{exc}}$ -- if it is non-empty -- in $\bigO(d \log_d n)$ time with high probability. The lemma then follows.
\end{proof}
Finally~\cref{thm:simplerand} follows by the above discussions by specifying the out-orientation algorithm. Choosing the one from Theorem~\ref{thm:exp_colouring_wc} yields the theorem.

\subsection{Bottlenecks}\label{s:bottlenecks}
In order to obtain an algorithm that maintains a colouring with $o(\alpha \log \alpha)$ colours, using \cref{alg:simple}, one would need to pick some number of random colours from palette of size $o(\alpha \log \alpha)$. No matter how we pick the palette size and the sample size, we either leave a vertex uncoloured with probability $\omega(1/\alpha)$ or we introduce high dependency between the outcomes of different vertices. 

In principle, one can handle some dependency between the random decisions made for vertices, as long as $\Delta \cdot p \leq (1-\eps) / \Delta$, where $\Delta$ is the maximum number of dependencies of a single vertices, $p$ is the probability of failure, and $\eps > 0$ is arbitrarily small constant. 

Nevertheless, the low-out degree orientation has no guarantees on in-degrees and the dependency degree $\Delta$ of \cref{alg:simple} can depend on $n$ rather than on $\alpha$. Given that whenever the number of colours we choose from is not significantly larger than the total number of colours proposed by the out-neighbours we introduce dependency between all vertices that share out-neighbours, using $o(\alpha \log \alpha)$ colours gives a failure probability that still only depends on $\alpha$ and a possibly unbounded dependency degree $\Delta$.

In \cref{s:better_colouring}, we propose a slightly different randomised partition that exploits the query ordering to control dependencies. More precisely, the ordering allows us to bound the dependency to $\Delta = \bigO(\poly(\alpha))$, while still providing that the probability of failure is such that $\Delta \cdot p \leq (1-\eps) / \Delta$, which we discuss in more detail in \cref{s:better_colouring}. 

\section{A better randomised colouring} \label{s:better_colouring}
In this section, we present a dynamic algorithm that provides a colouring with palette of colours that might depend on $n$ (or alternatively, requires a lower bound on the value of $\alpha$ to be at least some function of $n$). By combining the reduction to $\bigO(\log n)$ arboricity graphs from \cref{lem:reduction} with  \cref{thm:simplerand} one can obtain an algorithm that provides colouring with $\bigO(\alpha \log \log n)$ colours. 

Here, we show that it is possible to go slightly beyond this simple composition, and provide an algorithm that colours the vertices with $\bigO(\alpha \log \log \log n)$ colours. Together with the algorithm from \cref{thm:simplerand} it gives a dynamic algorithm that for a dynamic graph whose arboricity presently is $\alpha$, provides an implicit colouring with 
colours from a palette of size $\bigO{(\min \{\alpha \log \alpha, \alpha \log \log \log n \})}$.
\begin{theorem}\label{thm:better_colouring}
Let $G$ be a dynamic $n$-vertex graph undergoing edge insertions and deletions. There exists an algorithm that maintains an implicit $\bigO{(\min \{\alpha \log \alpha, \alpha \log \log \log n \})}$ vertex colouring, with update time $\bigO(\log^ 4 n)$ and query time $\bigO(\log^6 n)$, where $\alpha$ denotes the current arboricity of $G$. 
\end{theorem}

Note that this algorithm improves the state of the art for all ranges of $\alpha$. In particular in the regime where $\alpha$ is small, we obtain exponential improvements. In the regime where $\alpha = \poly \log n$, we obtain polynomial improvements and in the regime where $\alpha$ is large, we improve from $\bigO(\alpha \log n)$ colours to $\bigO{(\alpha \log \log \log \alpha)}$ which is the worst asymptotical improvement, we achieve. 

In the remainder of this section, we firstly introduce the algorithm that provides colouring with $\bigO(\alpha \log \log n)$ colours, and then we discuss ranges of parameter $\alpha$ in which it improves upon \cref{thm:simplerand}.

\paragraph{Algorithm outline}
Similarly to the algorithm using $\bigO(\alpha \log \alpha)$ colours described in \cref{s:simple_randomized}, we rely on a randomised reduction to graphs with arboricity $\bigO(\log n)$ stated in \cref{lem:reduction}. Thanks to this reduction, it is sufficient to provide an algorithm with update and query complexity $\poly(\alpha, \log n)$ in order to obtain an algorithm with $\poly( \log n$) update and query complexity. Also, it allows is to focus on graphs with arboricity $\bigO(\log n)$.  

The key technical contribution presented in this section is an algorithm based on a different random partitioning which allows us to further reduce the problem to one where the arboricity is at most $\bigO(\log^2 \log n)$.

To accomplish this, we follow the structure from \cref{s:simple_randomized} -- we maintain an $\bigO(\alpha)$ out degree orientation, and use it to design an algorithm with bounded volume. The small difference is that here, we rely on an acyclic orientation. To that end, we can use \cref{thm:acyclic}. As already mentioned at the end of \cref{s:simple_randomized}, an approach that tries to colour the vertices at random and colours them only when there is no collision with the out-neighbours cannot be used as it is, in order to get down to significantly below $\alpha\log \alpha$ colours.

Here, we propose an approach that uses \emph{arb-defective colourings}, that is, an algorithm that colours the vertices in a way that allows for a vertex to have a few colour collisions with its out-neighbours. The general idea is that as long as for each $v$ coloured with $s(v)$, the set of out-neighbours of $v$ contains only a few vertices coloured with $s(v)$, we can think about such colouring as a partition into several vertex disjoint problems of lower arboricity. Then, we can colour those lower arboricity problems using the simple algorithm with $\bigO{}(\alpha \log \alpha)$ colours. 

Unfortunately, our algorithm relies on random partition technique that cannot guarantee that all vertices have sufficiently few out-neighbours with colliding colour. We say that a vertex that has too many out-neighbours with colliding colour is \emph{conflicting}. Similarly as in \cref{s:simple_randomized}, we colour a \emph{conflicting} vertex $v$ with a colour from a separate palettes, by considering all \emph{conflicting} vertices that are reachable from $v$. In order to obtain an efficient algorithm, we need to provide a partition scheme that provides both small probability of being \emph{conflicted}, and a bound on the degree of dependency graph. We are able to limit the 
probability of a vertex being \emph{conflicted} by sampling colours from the range $\bigO(\alpha / \log^2 \alpha)$, and by allowing a sufficiently large number of collisions with the out-neighbours. Then, in order to obtain the desired upper bound on the volume of the query algorithm, we also use symmetry breaking that follows from the order of queries to the vertices. 

\subsection{Description of the algorithm} \label{sec:desc_alg}
As hinted to earlier, the idea is to partition the graph into subgraphs with much lower arboricity i.e.\ we wish to compute an arb-defective colouring. 
Given a graph with a $d$-bounded out degree orientation, the algorithm computes a random partition of the vertices into $\bigO(d /\log^2 d)$ vertex disjoint graphs each with arboricity bounded by $\bigO(\log^2 d)$ provided that $d = \Omega(\poly \log n)$. By applying the simple algorithm from Section~\ref{s:simple_randomized}, on the subgraphs obtained from this reduction, we then arrive at our final algorithm.

In order to decompose into these random subgraphs, we have to exclude some vertices from receiving a colour i.e.\ from joining a subgraph. Namely, some vertices might have too many neighbours in the subgraph, they proposed to join. Thus to ensure that the subgraphs have low arboricity, we have to exclude them from joining their proposed subgraph. 

The excluded vertices will form shattered into components, that we shall later show are small. This allows us to colour these vertices differently with a separate palette. To colour a specific excluded vertex, we determine the relevant part of its excluded component, and colour these in a greedy way similarly to what we did in Section~\ref{s:simple_randomized}.

\paragraph{Random partition}
To construct the random partition, we, similarly to Section~\ref{s:simple_randomized}, make vertices propose colours. The hope is for a vertex to have a few colour collisions with its out-neighbours. As elaborated on earlier, it seems unlikely that a simple proposition scheme a kin to that from Section~\ref{s:simple_randomized} will work. 
To overcome this, we take the order in which vertices propose colours into account. Intuitively, it should be easier to propose a colour with few collisions, if we already know the colour of some of the neighbouring vertices. 
And this is exactly what we will do. Hence in order to make further discussion more concise, we define three types of vertices: 
\begin{itemize}
	\item \emph{clear} vertices that are not coloured yet, and did not propose a colour, 
	\item \emph{prompted} vertices that already proposed a colour, but did not determine if they join a subgraph or become excluded.
	\item \emph{settled} vertices that already settled on a subgraph to join or determined that they are excluded. 
\end{itemize}
When a clear vertex needs to propose its colour, it will look at the colours present at its prompted and settled out-neighbours. If a colour here is present too often, it will remove this colour from the palette it proposes its colour from. To be more precise: initially a vertex $v$ with $d$ out-neighbours has the palette $S_v = [\frac{2(1+\delta)d}{p}]$ available to it, where $p = \ceil{\log^2 d}$ and $\delta \geq 3$ to be specified later. If a colour was proposed more than $p$ times at prompted or settled out-neighbours, the colour is removed from $S_v$. Note that in this way at most $\frac{d}{p}$ colours are removed, so any vertex will always have at least $(1+\delta)\frac{d}{p}$ colours available to it. 
After removing these colours, $v$ then proposes a colour uniformly at random from $S_v$. 

While this scheme might look reasonable at first glance, it actually has not helped much yet. The problem is that if many of the out-neighbours are clear vertices, then we have not actually gained anything over the scheme from Section~\ref{s:simple_randomized}. Indeed, clear vertices might have linearly many prompted in-neighbours that end up all depending on the proposed colour at $v$. Hence, again we have the problem that too many random experiments depend on each other. 

To overcome this. we do as follows: first let the \emph{badness}, $b(v)$, of a vertex $v$ be the number of prompted or settled in-neighbours of $v$: 
\[
b(v) = \begin{cases} |\{w \in N^{-}(v): \text{w is settled or prompted} \}| \text{  if $v$ is clear}\\
				0 \text{  otherwise}
				\end{cases}			
\]
We say an edge $u\rightarrow v$ is \emph{bad} if it contributes to $b(v)$ i.e.\ if $u$ is either prompted or settled and $v$ is clear.  If an edge is or once was \bad, say that it is \unfortunate. By ensuring that all clear vertices propose a colour when they have $2d$ \bad in-edges, we also ensure that no vertex has more than $2d$ \unfortunate in-edges. If two vertices are connected by \unfortunate out-edges to the same vertex, we say that $u$ and $v$ \depend on each other. In order to avoid the problem above, we need to limit the dependencies between vertices, and as just discussed this can be achieved by ensuring that no clear vertex $v$ has more than $2d$ settled or prompted in-neighbours i.e.\ that $b(v) \leq 2d$.

In order to limit $b(v) \leq 2d$ for all vertices $v\in V$, we have to adjust the order in which vertices propose colours. To this end, if we would like to know the colour proposition of a vertex $v$, we first calculate if proposing a colour at $v$ would cause the badness of some vertex $u$ to overflow $b(u) > 2d$. It this is the case, then we must propose a colour at $u$ before proposing a colour at $v$. Hence, we add all such $u$'s together with $v$ to a set $H$. Now we calculate, if proposing colours in $H$ would cause the badness of a new set of vertices to overflow. If so, we add these vertices to $H$. We repeat this procedure until no new vertices overflow. 
Then we propose the colours according to the reverse of the topological ordering of the vertices. This will ensure that no vertex $w$ in $H$ has a badness overflow of $b(w) > 2d$. Note that the fact that our out-orientation is acyclic is crucial for this to be possible. 

In \cref{alg:propose} we include pseudocode of an algorithm that is used by \cref{alg:better_colouring} to propose colours for vertices. 
In a nutshell, we first put all vertices that need to propose a colour into a set $H$. Then we calculate the vertices which would become incident to too many \bad in-edges, if the vertices in $H$ were to propose colours. We add these vertices to $H$, and repeat the process until proposing colours for the vertices in $H$, causes no clear vertices to have too many bad in-edges. Finally, we propose the colours of vertices in $H$ in reverse topological order.
In order for a vertex $v$ to propose a colour, it first discards all colours that are proposed too many times by the out-neighbours of $v$, then it proposes a colour uniformly at random from the set of remaining colours. The complexity of this algorithm is analysed in \cref{sec:analysis}.

\begin{algorithm}[H]
\KwIn{
$G \leftarrow $ input graph processed in this recursive call\\
$v \leftarrow $ a vertex to be coloured \\
$d \leftarrow $ an upper bound on the arboricity \\
$c \leftarrow $ id of an instance; $c \cdot d$ is the smallest colour reserved for this sub-problems residual graph \\
}
if $v$ already proposed a colour:
\ \ return $c(v)$
else
\ \ $H \gets v$ \\
\ \ $H' \gets v$ \\
\ \ while(H' is not empty) \\
\ \ \ \ $H' \gets \{u: b(u) + d^{-}_{G[H \cup H' \cup \{u\}]}(u) > 2d \}$ \\  
\ \ \ \ $H' \gets H \cup H'$ \\
\ \ $p \gets \ceil{\log^2 d}$\\
\ \ for all $u \in H$ in reverse topological order:\\
\ \ \ \ $S' \gets S \setminus$ colours proposed more than $p$ times by out-neighbours of $u$\\
\ \ \ \ $c(u) \gets $ colour selected uniformly at random from $S'$\\
\ \ \ \ increase $b(w)$ of all \emph{clear} out-neighbours $w\in N^{+}(u)$ \\

\caption{Propose Colour $(G, v, d)$}\label{alg:propose}
\end{algorithm}
Note that in each iteration of the while-loop, we set $H'$ to contain the neighbours that would cross the \emph{badness threshold} if the vertices in $H$ were to propose colours. In practice this can be detected by updating $d^{-}_{G[H \cup  \{u\}]}(u)$ as $H$ changes.

\paragraph{Colouring shattered components} We say that a vertex $v$ fails its colour experiment, if more than $2(1+\delta)p$ of its neighbours proposed the same colour as $v$.
When a vertex $v$ fails it colour experiment, it is excluded from entering the subgraph induced by the colour class of the colour, it proposed. In order to colour this vertex, we instead determine the set $R(v) \subset G_{\mathit{exc}}$ containing vertices with failed colour experiments that are reachable via directed paths -- only containing vertices with failed colour experiments -- from $v$. 
Then, we colour $G[R(v)]$ using $3d+1$ colours similarly to Algorithm~\ref{alg:colour_residual} and Lemma~\ref{lem:palette_for_residual} in Section~\ref{s:simple_randomized}: we compute a $2d$ degeneracy ordering of $G'[R(v)]$, and colour the vertices according to this and the original out-neighbours. 
That is a vertex picks a colour that is not present at any of its out-neighbours or any of the vertices preceding it in the degeneracy ordering.

In \cref{alg:colour_residual2} below, we provide pseudo-code of an algorithm that colours the residual graph $R(v)$, provided that $v$ is excluded in $G$. First we determine $R(v)$ in a BFS-like fashion. Note, however, that we could also determine $R(v)$ in a DFS-like fashion or according to any other way of traversing an implicitly represented graph. 

\begin{algorithm}[H]
\KwIn{
$v \leftarrow $ a vertex in $G'_{\mathit{exc}}$ to be coloured \\
$d \leftarrow $ an upper bound on the arboricity \\
}
$R(v) \gets v$ \\
$Q.$push$(v)$ \\
while(Q is not empty) \\
\ \ $u \gets Q.$pop() \\
\ \ for $w \in N^+(u)$ \\
\ \ \ \  $s(w) \gets$ ProposeColour$(w)$ \\
\ \ \ \  $\mathit{counter} \gets 0$ \\
\ \ \ \  for $t \in N^+(w)$ \\
\ \ \ \ \ \  if ProposeColour(t) = $s(w)$ \\
\ \ \ \ \ \ \ \ $\mathit{counter} \gets \mathit{counter} + 1$ \\
\ \ \ \ if $\mathit{counter} > 2(1+\delta)p$ \\
\ \ \ \ \ \ $Q.$push$(w)$ \\
\ \ \ \ \ \ $R(v) \gets R(v) \cup \{w\} $ \\
Colour $G[R(v)]$ as in Algorithm~\ref{alg:colour_residual}
\caption{Colour Residual $(G', v, d, c)$}\label{alg:colour_residual2}
\end{algorithm}

\paragraph{The query algorithm} In order to answer a query, we put all of the above together. We wish to use colouring scheme explained above to create the random partitions with arboricity $\bigO{(\log^2(d))}$, which we can then colour using the colouring algorithm from Section~\ref{s:simple_randomized}. Then we will use Algorithm~\ref{alg:colour_residual2} to colour the vertices excluded  from joining a subgraph. 

The idea is to do the execute the above from the view of a single queried vertex $v$. We will later show that with high probability the answer to the query to $v$ will only by affected by a volume bounded by $\bigO(\poly(d, \log n))$ and that all steps can be implemented in an efficient manner with respect to time-complexity. 

So to recap: when $v$ is queried, $v$ needs to determine whether or not it joins a subgraph or if it got put into to the excluded subgraph. Either way, once this has been determined $v$ will get coloured either by the algorithm from Section~\ref{s:simple_randomized} or by Algorithm~\ref{alg:colour_residual2}. 

In~\cref{alg:better_colouring}, we provide pseudo-code for an algorithm implementing the above. In order to answer a query, we call $\operatorname{Arboricity-Dependent Colouring} (G, v, d, 1)$ where $d$ is the maximum out-degree (or an upper bound of it) in the maintained out-orientation, where $v$ needs to propose a colour. 

\begin{algorithm}[H]
\KwIn{
$G \leftarrow $ input graph processed in this recursive call\\
$v \leftarrow $ a vertex to be coloured \\
$d \leftarrow $ an upper bound on the arboricity \\
$c \leftarrow $ id of palette to use in algorithm from Section~\ref{s:simple_randomized} \\
}

if $d < o(\poly \log n)$: \\
\ \ colour $v$ with the algorithm from Section~\ref{s:simple_randomized} using palette assigned to $c$. 
$p \gets \ceil{log^2 d}$
$c(v) \gets$ ProposeColour$(v)$
for $u \in$ clear out-neighbours of $v$:\\
\ \ ProposeColour$(u)$ 

if $\#$ of out-neighbours of $v$ that proposed $s(v) < 2(1+\delta)p$:\\
\ \ \ \ $G' \gets $ subgraphs of $G$ induced by vertices coloured with $c(v)$ \\
\atcp{$G'$ is computed in a lazy manner}

\ \ \ \ Arboricity-Dependent Colouring $(G', v, 2(1+\delta)p, c(v))$\\
else\\
\ \ \ \ Colour Residual $(G, v, d, c)$
\caption{Arboricity-Dependent Colouring $(v, d)$}\label{alg:better_colouring}
\end{algorithm}

\subsection{Analysis}\label{sec:analysis}
In this section, we will analyse the algorithm described in Section~\ref{sec:desc_alg}. There are many different components that needs to be analysed namely 1) the volume of a vertex in the shattered components, 2) the amortised complexity of the colour proposition mechanism, 3) the total time-complexity of performing a query, and finally 4) the number of colours employed by the algorithm.  
We will do this in the order described above beginning with the bound on the volume in the shattered components of excluded vertices. 

\paragraph{Analysis of volume in shattered components} First we briefly recall our overall strategy as outlined in the high-level overview earlier on. 
We first show that the probability that a vertex $v$ gets coloured by the colouring experiment is $O(1/\poly(d))$ by assuming that $v$ is in the worst possible scenario and then applying a Chernoff bound. 
The problem is that experiments at vertices sharing out-neighbours might be correlated -- and, furthermore, that the fact that the outcome of the experiments $x$ and $y$ are correlated, might even be correlated with the actual outcomes of the experiments at the vertices. 
To handle these problems, we introduce a graph $H$ which is sure to account for all dependencies that we cannot control, by adding an edge $uv$ if $u$ and $v$ share an out-neighbour $w$ and, furthermore, both of the edges $v \rightarrow w$ and $u \rightarrow w$ are \bad. 
The tricky part is that the exact structure of $H$ is heavily correlated with the outcomes at different vertices, and so we can only use some general properties of $H$ which are true regardless of the run of the algorithm. Due to the way colours are proposed, we are able to bound the maximum degree of $H$. This allows us to use an approach similar to the one used in Lemma~\ref{lem:shattering}.

First, we formally define $H$.
\begin{definition}
Let $H$ be a digraph on $V(G)$ such that $u \rightarrow v \in E(H)$ if and only if $u\rightarrow w, v \rightarrow w \in E(G)$ and furthermore $v \rightarrow w$ is \unfortunate, when the algorithm terminates. We call $H$ the \emph{dependency graph} of the algorithm, and note that it is a random graph whose distribution depends on both the query sequence and the outcome of the colour experiments performed by the algorithm.
\end{definition}
A first observation is that the maximum out-degree of $H$ at all times is bounded by $2d^2$:
\begin{observation} \label{obs:maxdeg}
No matter the run of the algorithm, we have $\max \limits_{v \in V(H)} d^+(v) \leq 2d^2$.
\end{observation}
\begin{proof}
Colours are prompted so that each vertex has at most $2d$ \unfortunate in-edges. This means that a vertex $v$ can share an out-neighbour with at most $2d \cdot{}d$ other vertices that points to said out-neighbour with an \unfortunate edge, and so we find that $d^{+}_{H}(v) \leq 2d^2$ for all $v$. 
\end{proof}
We will say a set of vertices in $V' \subset V(G)$ is \emph{semi-independent} if we have that for any edge $u\rightarrow v \in H[V']$, we must have that $v$ comes before $u$ in the acyclic ordering, meaning that there can be no directed path from $u$ to $v$ in $G$.
A consequence of this observation is the following: any connected subgraph of $G$ will contain a tree such that a relatively large subset of vertices from this tree form a semi-independent set in $H$. We call such a tree a \emph{marked tree}. More formally we have:
\begin{definition} \label{def:tm}
A rooted tree $T, \subset G$ is \emph{$t$-marked} if all of the following hold:
\begin{enumerate}
	\item it contains $t$ marked vertices $m_1, \ldots, m_t$.
	\item the root $r$ is marked.
	\item the marked vertices form a semi-independent set in $H[V(T)]$ containing the root.
	\item all non-marked vertices in $T$ is pointed to by a marked vertex in $H$.
	\item the vertices can be ordered so that $m_1 = r$ and for all $i$ the minimum subtree containing $m_1, \ldots, m_i$ is an $i$-marked 
	tree not containing any of the vertices $m_{i+1}, \ldots, m_t$.
\end{enumerate}
\end{definition}
\begin{definition}
Two $t$-marked trees $T_1$ and $T_2$ are said to be \emph{near-identical} if the marked vertices in $T_i$ $m_{1,i}, \ldots, m_{t,i}$ can be ordered according to condition 5) in Definition~\ref{def:tm} such that for all $j$ $m_{j,1} = m_{j,2}$.
\end{definition}
\begin{remark}
We let $NI(T)$ denote the set of $t$-marked trees that are near-identical to $T$.
\end{remark}
\begin{lemma} \label{lem:cont_t_tree}
Let $S \subset G$ be a subgraph containing $|S|$ vertices such that $S$ has a source which can reach all vertices of $S$ via directed paths. Then $S$ contains at least one $t$-marked tree where $t = \ceil{\frac{|S|}{2d^2+1}}$. 
\end{lemma}
\begin{proof}
We will construct the $t$-marked tree greedily. Order the vertices of $S$ according to the acyclic ordering, and pick the vertex highest in the ordering as the root. Note that this vertex necessarily is the source. Having picked $i$ marked vertices, we pick the $(i+1)$'th marked vertex as follows:
by Observation~\ref{obs:maxdeg}, H has maximum out-degree at most $2 \delta ^ 2$. This means that at most $i(2 \delta ^ 2)$ vertices cannot be picked as the $(i+1)$'th marked vertex, because they are connected by an in-edge to a marked vertex in $H$. Consider all of the vertices reachable via directed paths beginning at a vertex already in $T$ such that every internal vertex is either marked or connected to a marked vertex in $H$. Pick such a vertex that comes first in the acyclic ordering as the $(i+1)$'th marked vertex and extend $T$ to contain this vertex by adding some of the edges from a directed path reaching the vertex without creating a cycle. 
Notice that once we add a new marked vertex, we will never add a new marked vertex before it in the acyclic ordering, and so the set of marked vertices form a semi-independent set by construction.
As long as $i(2 \delta ^ 2+1) \leq |S|$, we may extend $T$ to contain another marked vertex. Notice that by construction, all of the above conditions are satisfied.
\end{proof}
One subtle difficulty is that as soon as one locks down too much of the structure of a $t$-marked tree, one also implicitly says something about the run of the algorithm and hence the outcomes of the colour experiments at the vertices of the tree. To overcome this, we will, as we do with $H$, work with $t$-marked trees as random graphs and only use properties of such trees that we know will hold, no matter the run of the algorithm. A very useful property of $t$-marked trees is that we know that no matter the run of the algorithm, the outcomes at a semi-independent set of vertices in $T$ will be based on disjoint colour experiments for each vertex. This allows us to control any (small) remaining dependencies arising from query order. Hence, we find:
\begin{lemma} \label{lem:t_tree}
Let $T$ be a random $t$-marked tree such that $m_1, m_2, \dots, m_t \in T$ form an independent set in $H$. Then 
\[
P\paren{\exists T' \in NI(T) \text{ that is uncoloured}} \leq \paren{\frac{1}{d^{\delta}}}^{t}
\]
provided that $\delta \geq 3 $. 
\end{lemma}
\begin{proof}
We may assume without loss of generality, possibly by renumbering, that $m_i$ proposed its colour as number $i$ vertex among $m_1, m_2, \dots, m_t$. Now if we can show that for all $1 \leq j \leq t$ 
\[
P\paren{m_j \text{ fails col. exp. } | m_i \text{ fails col. exp. } \forall i < j} \leq \frac{1}{d^{\delta}}
\]
then we may write 
\begin{align*}
P\paren{\exists T' \in NI(T) \text{ that is uncoloured}} \leq & P\paren{m_i \text{ fails col. exp. } \forall i}\\ = & P\paren{m_t \text{ fails col. exp. } | m_i \text{ fails col. exp. } \forall i < t}\cdot{} \\
&P\paren{m_{t-1} \text{ fails col. exp. } | m_i \text{ fails col. exp. }  \forall i< t-1}\cdot{} \ldots \cdot{} \\
& P\paren{m_1 \text{ fails col. exp.}}\\  \leq &  \paren{\frac{1}{d^{\delta}}}^{t}
\end{align*}
So we will focus on proving this statement. To this end, fix $j$ so that $1 \leq j \leq t$. Consider the experiment at $m_j$. The worst possible setup for the experiment at $m_j$ is that 1) $m_j$ has exactly $d$ out-neighbours, 2) every out-neighbour of $m_j$ was uncoloured when $m_j$ proposed its colour, 3) that both $m_j$ and its out-neighbours proposed colours from the same palette of colours and that this palette has minimum size and 4) $m_j$ already sees $p$ out-neighbours of each colour. Of course all of these events cannot happen simultaneously, but since we do not know the exact setup, we will assume the above. This can only lower the success-rate of an experiment. 

Initially, a palette has size $2(1+\delta)\frac{d}{p}$. Now, if some colour is present more than $p$ times at an out-neighbour, we remove that colour from the palette. Hence, we may remove at most $\frac{d}{p}$ colours from a points palette. Therefore, the minimum size of a palette is at least $(1+\delta)\frac{d}{p}$. Since we assume that $m_j$ sees $p$ out-neighbours of each colour, the colour experiment will fail if $(1+\delta)\frac{d}{p}$ of $m_j$'s out-neighbours that proposes a colour after $m_j$ proposes the same colour as $m_j$. Let $X_{m_{j}}$ indicate this event.

Since the $m_i$ are semi-independent in $H$, the outcome at $m_i$ for $i < j$ does not depend on the colours at the out-neighbours of $m_j$. 
Also the in-edges of $m_j$ in $H$ are determined before $m_j$ proposes its colour, and so its outcome is not correlated with its in-edges in $H$.
This means that neither the outcomes at $m_i$ nor the fact that the $m_i$ form a $t$-marked tree can influence the distribution of the colour experiments at the out-neighbours of $m_j$ more than we have already accounted for. 
Since such an out-neighbour $v$ proposes a colour uniformly at random from its palette, the event that it chooses the same colour as $m_j$ can be stochastically dominated by an independent random variable $I_{v}$ that is equal to one with probability at most  $\paren{(1+\delta)\frac{d}{p}}^{-1}$. Note that then the indicator variables $\{I_v\}_{v \in N^{+}(m_j)}$ are independent. 
Now $X_{m_{j}}$ can be stochastically dominated by the indicator variable $Y_{m_{j}}$ that indicates the event that  $\sum \limits_{v \in N^{+}(G)} I_v > (1+\delta)\frac{d}{p}$. Since the indicator variables $\{I_v\}_{v \in N^{+}(m_j)}$ are independent, we may apply a Chernoff bound to estimate the probability that $Y_{m_{j}} = 1$: 
\begin{align*}
P(X_{m_j} = 1 | m_i \text{ fails col. exp. } \forall i < j) &\leq P(Y_{m_j} = 1 | m_i \text{ fails col. exp. } \forall i < j) \\
& \leq P(\sum \limits_{v \in N^{+}(G)} I_v  > (1+\delta) p)
\end{align*}
Since 
\[
\mathbb{E}\paren{\sum \limits_{v \in N^{+}(G)} I_v} = d \cdot{} \paren{(1+\delta)\frac{d}{p}}^{-1} = \frac{p}{1+\delta}
\]
applying a Chernoff bound yields:
\[
P(X_{m_j} = 1 | m_i \text{ fails col. exp. } \forall i < j) \leq \exp\paren{-\frac{(1+\delta)\delta^2}{3(1+\delta)}p} \leq  \exp(-\delta \cdot p) \leq d^{-\delta}
\]
where we used the fact that $\delta \geq 3$ and that: 
$$(1+\delta)p = (1+\delta)^2 \cdot{} \mathbb{E}\paren{\sum \limits_{v \in N^{+}(G)} I_v} \geq  (1+\delta(1+\delta)) \cdot{} \mathbb{E}\paren{\sum \limits_{v \in N^{+}(G)} I_v} $$
\end{proof}
If we let $R(v)$ be the set of vertices reachable from $v$ in $G_{\mathit{exc}}$. Then we now find the following upper bound on the size of $R(v)$ with high probability:
\begin{theorem}
Suppose $\delta \geq 9 + \log c + l $ and that $d^l \geq \log_{d} n$. Then $ |R(v)| \leq c \cdot{}(2d^2+1) \log_{d} n $ with probability at least $1-n^{-c}$. 
\end{theorem}
\begin{proof}
We know by Lemma~\ref{lem:t_tree} that the probability that a $(c\log_{d} n)$-marked tree is uncoloured is very small. If we can show that no such tree is uncoloured with high probability, we then reach a contradiction to $|R(v)| > c \cdot{}(2d^2+1) \log_{d} n$, since then by Lemma~\ref{lem:cont_t_tree} $R(V)$ contains an uncoloured $(c\log_{d} n)$-marked tree no matter the run of the algorithm. Notice that by  Lemma~\ref{lem:t_tree} once we consider a tree $T$, we don't have to also consider $T' \in NI(T)$. Let $\tau$ be the minimum number of $t$-marked trees needed so that there exists $\tau$ $t$-marked trees $T_1, \ldots, T_{\tau}$ such that for any $t$-marked tree $T$, we have $T \in NI(T_i)$ for some $i$.

We first bound $\tau$ for any run of the algorithm. 
\begin{claim}
No matter the run of the algorithm
\[
\tau n\cdot{} 4^{t} \cdot{} (3d^3)^t \cdot{} t!
\]
 when all vertices have proposed a colour. 
\end{claim}
\begin{proof}
We have at most $4^t$ rooted, unlabeled trees of size $t$. If we contract edges where at least one endpoint is not marked in a $t$-marked tree, we must end up with one such rooted, unlabeled tree of size $t$. If we order such edges according to first the endpoint highest in the acyclic ordering, and settling ties by the order of the second endpoint, each $t$-marked tree contracts to a unique rooted, unlabeled tree of size $t$, by always contracting the highest order edge with at least one unmarked neighbour. Therefore, we just need to bound the number of ways, we can embed a $t$-marked tree such that it contracts to a specific rooted, unlabeled tree of size $t$. Once we consider a $t$-marked tree $T$, we don't have to consider any tree in $NI(T)$. 

We will embed the vertices according to the ordering given by condition 5) in Definition~\ref{def:tm}. Note first that when we wish to embed the $m_i$, we have at most $3d^{3}i$ choices for how to embed it. Indeed, a vertex from the semi-independent set has out-degree at most $2d^2$ in $H$, and its out-neighbours together with the out-neighbours of its neighbours in $H$ sum to $2d^2 \cdot{} d + d\leq 3d^{3}$. Hence, when we wish to embed $m_i$, we have at most $3d^{3}i$ vertices reachable via paths $P$ where every internal vertex of $P$ is pointed to by at least one marked vertex in $H$. Which edges we choose to connect $m_i$ to $T$ via, does not matter, since no matter the choices, by construction, we end up with trees that are near-identical to each other. This follows from the fact that only the choice of marked vertices determines the next choices for the next marked vertex, not the tree spanning the marked vertices. 

Since we have at most $n$ choices for a root, we find that $\tau$ for any run of the algorithm is upper bounded by 
\[
\tau \leq n\cdot{} 4^{t} \cdot{} (3d^3)\cdot{} \ldots \cdot{} (3td^3) = n\cdot{} 4^{t} \cdot{} (3d^3)^t \cdot{} t!
\]
\end{proof}
For $t = c \log_{d} n$, we find that any run of the algorithm we only have to consider 
\begin{align*}
n\cdot{} 4^{c \log_{d} n} \cdot{} (3d^3)^{c\log_{d} n} \cdot{} (c\log_{d} (n))! &\leq n\cdot{} 4^{c \log_{d} n} \cdot{} (3d^3)^{c\log_{d} n} \cdot{} (c\log_{d} (n))^{(c\log_{d} (n))} \\
&\leq n^c \cdot{} n^{2c} \cdot{} n^{5c} \cdot{} (cd^{l})^{c \log_d n}  \\
&\leq n^{c(8+\log(c)+l)}
\end{align*}
$t$-marked trees in order for us to consider at least one near-identical $t$-marked tree for each $t$-marked tree left by the algorithm. By Lemma~\ref{lem:t_tree}, the probability that one such $t$-marked tree -- or a $t$-marked tree that is near-identical to it -- survive is at most 
\[
\paren{d^{-\delta}}^{c \log_{d} n} \leq n^{-c\delta}
\]
and hence by a union bound, the probability that no $t$-marked trees survive can be upper bounded by
\[
n^{c(8+\log(c)+l)} \cdot{} n^{-c\delta} = n^{-c(\delta-(8+\log(c)+l))} \leq n^{-c}
\]
for $\delta \geq 9 + \log(c) + l$. 
It then follows by Lemma~\ref{lem:cont_t_tree} that if $ |R(v)| > c \cdot{}(2d^2+1) \log_{d} n $, then it must contain an uncoloured $t$-marked tree, so this happens with probability at most $n^{-c}$ by above, which is what we wanted to show.
\end{proof}

\paragraph{Proposing colours}
Here, we analyse the complexity of the colour proposing algorithm that helps us reduce the dependencies between the events of including particular vertices to $G_{\mathit{exc}}$. We rely on the following observation: if a vertex $v$ has an out-neighbour $u$ that has already proposed its colour, then $v$ can take $u$'s colour-proposition into account, when proposing its own colour. In particular, $v$'s decision to remain uncoloured can be made to (almost entirely) depend on the proposed colours of vertices the propose a colour \emph{after} $v$ proposes its own colour i.e.\ $v$'s decision depends only on vertices to whom it points to via an \unfortunate edge. In particular, we can control the dependencies by maintaining the invariant that no vertex has more than $2d$ \unfortunate in-edges. 

We show that we can charge the cost of the colour proposal scheme to the number of vertices that the algorithm has to visit in order to answer a query. In order to analyse the cost of the colour-proposition scheme, we can charge the cost of proposing colours to the event that the algorithm visits a vertex in order to determine, if it is uncoloured. Say a visit is \emph{\necessary} if the algorithm visits a vertex $v$ to propose a colour at $v$ because we need to know the outcome of the experiment at $v$ or at an in-neighbour of v. Note that because the subgraphs are only given implicitly, we sometimes need to propose colours in the arb-defective colouring scheme in order to determine for example the shattered volume in a subgraph, where the simple algorithm is run. The visits incurred like this are also \emph{\necessary}. On the other hand, the colour proposals made by the algorithm simply to balance the number of \bad in-edges are said to be \emph{\unmotivated}. 

First, we show that we do not need to make too many colour proposals for the arb-defective colouring in order to make the colour proposals needed at the second level in the subgraphs. 
\begin{proposition}
\label{prop:lvli}
For $v$ to determine the outcome of its colour experiment at the second level, we need only make $d$ \necessary colour proposals on the first level.
\end{proposition}
\begin{proof}
In order for $v$ to determine the outcome of its colour experiment at the second level, it needs to know which of its out-neighbours that might join the same subgraph as $v$ on the second level. To do so, it needs to know the colour of no more than its $d$ out-neighbours on the first level.
\end{proof}

Using Proposition~\ref{prop:lvli}, we may now bound the total number of colour proposals as a function of the number of \necessary visits:
\begin{lemma} \label{lem:amt}
Suppose that the algorithm has needed to determine the outcome of the colour experiment of at most $\gamma$ vertices across the two levels. Then at most $\bigO{(\gamma\cdot{}d)}$ vertices have proposed a colour on the first level. 
\end{lemma}
\begin{proof}
Recall that an edge $u \rightarrow v$ is \bad if $u$ proposed a colour before $v$ and $v$ has not proposed a colour yet, and \unfortunate if $u$ proposed a colour before $v$ and $v$ subsequently proposed a colour . Define a potential function $\Phi$ to be the number of \bad edges in $G$:
\[
\Phi(G) := \textit{No.\ of \bad edges on the first level}
\]
When a vertex proposes a colour at the first level, at most $d$ new \bad edges are created, and so $\Phi$ increases by at most $d$.
When a vertex $u$ proposes a colour at the second level at most $d^2$ new \unfortunate edges are created at the first level by Proposition~\ref{prop:lvli}. 
Thus $0 \leq \Phi(G) \leq \gamma \cdot{}d^2$. Since proposing a colour at a vertex with $2d$ \bad in-edges turns all of these edges \unfortunate instead of \bad, this drops the potential by exactly $2d$, it follows that the $\Phi$ drops by $2d$ every time we make an \unmotivated colour proposal. Hence, in total at most $\frac{\gamma \cdot{}d^2}{2d} = \bigO{(\gamma\cdot{}d)}$ \unmotivated colour proposals have been made. Since determining the outcome of the colour experiment at $\gamma$ vertices requires at most $\bigO(\gamma\cdot{}d)$ \necessary visits on the first level by Proposition~\ref{prop:lvli}, the lemma follows by summing the two contributions from above. 
\end{proof}

In particular, we know that if we need to know the outcome of the colour experiments of $\gamma$ vertices in total across the two levels, then we have performed at most $\bigO{(\gamma\cdot{}d)}$ colour proposals on the first level. 

\paragraph{Query time and number of colours} 
The following lemma analyses the query time of~\cref{alg:better_colouring}. 
\begin{lemma} \label{lem:q}
The query of \cref{alg:better_colouring} takes amortised $O((d^4 \log n))$ time with probability at least $1-n^{c-2}$, provided that $\delta \geq 9 + \log c + l$ and $d^l \geq \log_d n$.
\end{lemma}
\begin{proof}
By union bounding over overall vertices $v$, $R(v)$ has size at most $|R(v)| = \bigO{(d^2 \log n)}$ on both levels with probability at least $1-n^{c-2}$ for all vertices. Conditioned on this event, we need to know the outcome of $\bigO{(t\cdot{}d^2 \log n)}$ colour experiments to answer $t$ queries. 
By Lemma~\ref{lem:amt}, this means that we make at most $\bigO{(t\cdot{}d^3 \log n)}$ colour proposals in total across both levels. A colour proposal can be performed in $\bigO{(d)}$ time, so in total we spend $\bigO{(t\cdot{}d^4 \log n)}$ time proposing colours in order to answer $t$ queries. 
Finally, we spend at most $\bigO{(d^2 \log n)}$ time in order to colour each residual graph by Lemma~\ref{lem:palette_for_residual}, and at most $\bigO{(d^3 \log n)}$ finding them, if we disregard the time spent for proposing colours, which we analysed above. Since $t$ queries can induce the colouring of at most $t$ residual graphs, this means that we find an amortised query time of $\bigO{(d^4 \log n)}$ with probability at least $1-n^{c-2}$.
\end{proof}
Finally, we can count the number of colours used by the 2 level scheme. We partition the graph into $\bigO{(\delta \cdot \frac{d}{\log^2 d})}$ sub graphs, each of which have arboricity at most $\bigO{(\log^2 d)}$ and use at most $\bigO{(d)}$ colours to colour the residual graphs needed to do so. Each subgraph we can colour with $\bigO{(\log^2(d) \log \log d)}$ colours using the algorithm from Section~\ref{s:simple_randomized}. Hence the total no. of colours used is upper bounded by  $\bigO{(\delta \cdot d \cdot{} \log \log d)}$. Hence, we have shown:
\begin{lemma} \label{No. of colours.}
Provided that $\delta \geq 9 + \log c + l$ and $d^l \geq \log_d n$, \cref{alg:better_colouring} uses $\bigO{(\delta \cdot d \cdot{} \log \log d)}$ colours.
\end{lemma}

\subsection{Combining everything}
Finally, we will show how to combine everything to get Theorem~\ref{thm:better_colouring}. The idea is to track a constant approximation of $\alpha$ via Theorem~\ref{thm:exp_colouring_wc}. Then we will apply different query algorithms depending on the range of $\alpha$. If $\alpha = \Omega(\poly{\log n})$, we will first apply the reduction in Lemma~\ref{lem:reduction} to obtain $\bigO{(\frac{\alpha}{\log n})}$ graphs each with arboricity $\bigO{(\log n)}$ with high probability. Then we can apply the algorithm from above with $\delta = 10 + \log c$ to colour these graphs with $\bigO{(\log n \log \log \log n)}$ colours to get a colouring using $\bigO{(\alpha \log \log \log n)}$ colours. In the regime where $\alpha = \bigO{(\log^{\frac{1}{l}} n)}$, we apply the algorithm from this section with $\delta = 9 + \log c + l$ to get an $\bigO{(\alpha \log \log \alpha)}$ colouring. 
We can use the scheduling of Sawlani \& Wang~\cite{DBLP:conf/stoc/SawlaniW20} as described in~\cite{DBLP:conf/icalp/CR22} in order to ensure that we adapt to the current arboricity. Finally, for the remaining values of $\alpha$ we apply the query algorithm from Section~\ref{s:simple_randomized} directly in order to get an $O(\alpha \log \alpha)$ colouring of the graph. 

We maintain the structures needed for all three ranges at all times. Since we can maintain an acyclic out-orientation of a graph with arboricity at most $\bigO{(\log n)}$ in $\bigO{(\log^3 n)}$ update time by Theorem~\ref{thm:acyclic}, and an $\bigO{(\alpha)}$ orientation in  $\bigO{(\log^3 n)}$ update time via the algorithm from~\cite{AOOA-CHHRS} (i.e. Theorem~\ref{thm:approximatealpha}), it now follows by Lemma~\ref{lem:reduction} that the total update time is at most: 
$$\bigO{(\log^3 n + \log^3 n + \log^4 n)} =\bigO{(\log^4 n)}.$$ 
By Lemma~\ref{lem:reduction} and Lemma~\ref{lem:q} the query time is $\bigO(\log^6 n)$ with high probability. Hence, we arrive at Theorem~\ref{thm:better_colouring}.

\appendix

\section{Randomised Reduction to Low Arboricity Graphs}\label{a:to_low_arboricity}

In this section, we provide a reduction of arboricity dependent dynamic vertex colouring of general graphs, to arboricity dependent dynamic vertex colouring of graphs with arboricity $\bigO(\log n)$. That is we prove \cref{lem:reduction}.

\begin{lemma}[\cref{lem:reduction}]
Assume that there exists an algorithm $A$ that maintains a $f(\alpha_{\max})$ colouring of the graph in $U_A(\alpha_{\max}, \log n)$ update and $Q_A(\alpha_{\max}, \log n)$ query time, where $\alpha_{\max}$ is a known upper bound on the arboricity of the dynamic graph. Then, there exists an algorithm $A'$ that maintains a colouring with $(1+\eps)\lceil \frac{\alpha}{\log n} 8f((c+1)\log n) \rceil$ colours, in $U_A(\alpha, \log n)$ update time with probability at least $1-n^{-c^2/3+3}$, and $Q_A(\bigO(\log n), \log n)$ query time, where $\alpha$ is the current arboricity of the graph.
\end{lemma}
\begin{proof}
Given a graph of sufficiently large arboricity $\alpha \in \Omega(\log n)$, we can partition its vertices into $y \in 8\alpha / \log n$ sets $V_1, V_2, \dots, V_y$ uniformly at random. Then, the arboricity of a graph induced by $V_i$, for all $i$, is $\bigO(\log n)$, with high probability. 

Indeed; fix a vertex $u$ and let $V(u) = V_i$ be the set that $u$ belongs to that is $u \in V_i$. Now pick any graph $H$ with vertex set $V$ and endow it with any $8\alpha$ bounded out-orientation. Note that $G$ together with any $8\alpha$ out-orientation would be one such valid choice. Now consider the out-degree of $u$ in $H[V(u)]$. Observe first that we may assume that $u$ has $8\alpha$ out-neighbours -- otherwise $d^{+}_{G[V_i]}(u)$ is stochastically dominated by the random variable obtained by artificially adding the remaining out-neighbours. We find that 
\[
\mathbb{E}[d^{+}_{G[V_i]}(u)] = \mathbb{E}(\sum \limits_{v \in N^{+}(u)} [v \in V_i]) = \sum \limits_{v \in N^{+}(u)} P(v \in V_i) = 8\alpha \cdot{}\frac{\log n}{8\alpha} = \log n
\] 
Since each vertex is partitioned independently, we can apply a standard Chernoff bound to find that:
\[
P(d^{+}_{G[V_i]}(u) > (c+1)\log{n}) \leq \exp{(-\frac{c^2}{3}\log n)} = n^{-c^2/3}
\]
So the probability that $u$ has at most $(c+1)\log n$ out-neighbours in this particular graph is at most $n^{-c^2/3}$. Now we may union bound over all n vertices in $V$, and over $n^2$ choices of $H$ together with the out-orientation to find that with probability at least $1-n^{-c^2/3+3}$ all sets $V_i$ induce graphs with arboricity $\bigO{\log n}$ for all the $n^2$ choices of $H$. In particular, we may let our dynamic graph $G_{i}, \dots, G_{n^2+i}$ together with the maintained out-orientation play the role of $H$ and bounded out-orientation of $H$. Note that we are using the oblivious adversary assumption to arrive at this conclusion. 

Hence, with high probability our vertex partition works for $n^2$ updates in a row. This is more than sufficient to completely rebuild a new random partition and the required data structures in each partition in the back ground. Using standard techniques this can be deamortised. 

The only remaining problem is that $\alpha$ might change throughout the updates. But here we can use the algorithm of Sawlani \& Wang~\cite{DBLP:conf/stoc/SawlaniW20} to schedule the updates to $O(\log n)$ copies of $G$ in such a manner that we at all times have a pointer to a fully updated copy and an $\bigO(1)$ approximation of the current arboricity.
More precisely, Sawlani \& Wang show how to maintain copies with an estimated arboricity of $y_i = \alpha_i / \log n$ for arboricities $\alpha_i = 2^i$ for $i \in [\ceil{\log n}]$ as well as a pointer to a copy with an $\bigO(1)$ approximation of the current arboricity that is fully updated. This can be done in worst-case update time $\bigO(log^{4} n)$. 
For each copy of $G$ we maintain the above random partition. 

Let $G_{i,j}$ be the graph induced by $V_j$ in the partition into $y_i$ sets. When we actually insert or delete an edge from the graph for a particular partition, we verify whether both endpoints of an edge $\set{u,v}$ are in the same random subset, i.e. we check whether $V(u) = V(v)$. If that is the case, we insert the edge into to the graph $G_{i,j}$, for $j$ such that $V_j = V(u) = V(v)$ in the $i$'th partition. Overall, we need to perform edge insertion for $\bigO(\log n)$ different values of $i$, hence the overall worst-case cost is $\bigO(\log n) \cdot Q_A(\bigO(\log n), \log n)$.

Now in order to answer a query, we once again rely on the fact that we can track $\bigO(1)$ approximation of the current arboricity and that we have a pointer to a graph where the estimate realises this approximation. Whenever we get a query, we query the relevant copy, or the original graph if the current arboricity was to small to guarantee that the partition induces subgraphs with arboricity $\Theta(\log n)$. Given this partition, the colour of a vertex $u$ is $(j,c)$, where $j$ is the index such that $V_j = V(u)$, and $c$ is the colour of $u$ in the graph induced by $V_j$. Given that the arboricity of this graph is $\bigO(\log n)$, the overall query complexity is at most as large as the query complexity on a graph with arboricity $\bigO(\log n)$.
\end{proof}

\bibliographystyle{alpha} 
\bibliography{ref}

\newcommand{\etalchar}[1]{$^{#1}$}
\begin{thebibliography}{CHvdH{\etalchar{+}}22}

\bibitem[Bar16]{Barenboim16}
Leonid Barenboim.
\newblock Deterministic ({\(\Delta\)} + 1)-coloring in sublinear (in
  {\(\Delta\)}) time in static, dynamic, and faulty networks.
\newblock {\em J. {ACM}}, 63(5):47:1--47:22, 2016.

\bibitem[BB17]{berglinetal:LIPIcs:2017:8263}
Edvin Berglin and Gerth~Stolting Brodal.
\newblock {A Simple Greedy Algorithm for Dynamic Graph Orientation}.
\newblock In {\em 28th International Symposium on Algorithms and Computation
  (ISAAC 2017)}, volume~92 of {\em Leibniz International Proceedings in
  Informatics (LIPIcs)}, pages 12:1--12:12, Dagstuhl, Germany, 2017. Schloss
  Dagstuhl--Leibniz-Zentrum fuer Informatik.

\bibitem[BBKO22]{Balliu0KO22}
Alkida Balliu, Sebastian Brandt, Fabian Kuhn, and Dennis Olivetti.
\newblock Distributed {$\Delta$}-coloring plays hide-and-seek.
\newblock In Stefano Leonardi and Anupam Gupta, editors, {\em {STOC} '22: 54th
  Annual {ACM} {SIGACT} Symposium on Theory of Computing, Rome, Italy, June 20
  - 24, 2022}, pages 464--477. {ACM}, 2022.

\bibitem[BCHN18]{DBLP:conf/soda/BhattacharyaCHN18}
Sayan Bhattacharya, Deeparnab Chakrabarty, Monika Henzinger, and Danupon
  Nanongkai.
\newblock Dynamic algorithms for graph coloring.
\newblock In Artur Czumaj, editor, {\em Proceedings of the Twenty-Ninth Annual
  {ACM-SIAM} Symposium on Discrete Algorithms, {SODA} 2018, New Orleans, LA,
  USA, January 7-10, 2018}, pages 1--20. {SIAM}, 2018.

\bibitem[BCK{\etalchar{+}}19]{DBLP:journals/algorithmica/BarbaCKLRRV19}
Luis Barba, Jean Cardinal, Matias Korman, Stefan Langerman, Andr{\'{e}} van
  Renssen, Marcel Roeloffzen, and Sander Verdonschot.
\newblock Dynamic graph coloring.
\newblock {\em Algorithmica}, 81(4):1319--1341, 2019.

\bibitem[BDF{\etalchar{+}}20]{DBLP:conf/swat/BosekDFPZ20}
Bartlomiej Bosek, Yann Disser, Andreas~Emil Feldmann, Jakub Pawlewicz, and Anna
  Zych{-}Pawlewicz.
\newblock Recoloring interval graphs with limited recourse budget.
\newblock In Susanne Albers, editor, {\em 17th Scandinavian Symposium and
  Workshops on Algorithm Theory, {SWAT} 2020, June 22-24, 2020, T{\'{o}}rshavn,
  Faroe Islands}, volume 162 of {\em LIPIcs}, pages 17:1--17:23. Schloss
  Dagstuhl - Leibniz-Zentrum f{\"{u}}r Informatik, 2020.

\bibitem[BE10]{DBLP:journals/dc/BarenboimE10}
Leonid Barenboim and Michael Elkin.
\newblock Sublogarithmic distributed {MIS} algorithm for sparse graphs using
  {N}ash-{W}illiams decomposition.
\newblock {\em Distributed Comput.}, 22(5-6):363--379, 2010.

\bibitem[BEG22]{BarenboimEG22}
Leonid Barenboim, Michael Elkin, and Uri Goldenberg.
\newblock Locally-iterative distributed ({\(\Delta\)} + 1)-coloring and
  applications.
\newblock {\em J. {ACM}}, 69(1):5:1--5:26, 2022.

\bibitem[BEK14]{BarenboimEK14}
Leonid Barenboim, Michael Elkin, and Fabian Kuhn.
\newblock Distributed (delta+1)-coloring in linear (in delta) time.
\newblock {\em {SIAM} J. Comput.}, 43(1):72--95, 2014.

\bibitem[BEPS12]{DBLP:conf/focs/BarenboimEPS12}
Leonid Barenboim, Michael Elkin, Seth Pettie, and Johannes Schneider.
\newblock The locality of distributed symmetry breaking.
\newblock In {\em FOCS}, pages 321--330. {IEEE} Computer Society, 2012.

\bibitem[BF99]{Brodal99dynamicrepresentations}
Gerth~Stolting Brodal and Rolf Fagerberg.
\newblock Dynamic representations of sparse graphs.
\newblock In {\em In Proc. 6th International Workshop on Algorithms and Data
  Structures (WADS)}, pages 342--351. Springer-Verlag, 1999.

\bibitem[BF20]{blumenstock2019constructive}
Markus Blumenstock and Frank Fischer.
\newblock A constructive arboricity approximation scheme.
\newblock In {\em {SOFSEM} 2020: Theory and Practice of Computer Science - 46th
  International Conference on Current Trends in Theory and Practice of
  Informatics, {SOFSEM} 2020, Limassol, Cyprus, January 20-24, 2020,
  Proceedings}, volume 12011 of {\em Lecture Notes in Computer Science}, pages
  51--63. Springer, 2020.

\bibitem[BGK{\etalchar{+}}22]{DBLP:journals/talg/BhattacharyaGKL22}
Sayan Bhattacharya, Fabrizio Grandoni, Janardhan Kulkarni, Quanquan~C. Liu, and
  Shay Solomon.
\newblock Fully dynamic ({\(\Delta\)}+1)-coloring in \emph{O}(1) update time.
\newblock {\em {ACM} Trans. Algorithms}, 18(2):10:1--10:25, 2022.

\bibitem[BKK{\etalchar{+}}21]{benderESA21}
Michael~A. Bender, Tsvi Kopelowitz, William Kuszmaul, Ely Porat, and Clifford
  Stein.
\newblock {Incremental Edge Orientation in Forests}.
\newblock In Petra Mutzel, Rasmus Pagh, and Grzegorz Herman, editors, {\em 29th
  Annual European Symposium on Algorithms (ESA 2021)}, volume 204 of {\em
  Leibniz International Proceedings in Informatics (LIPIcs)}, pages
  12:1--12:18, Dagstuhl, Germany, 2021. Schloss Dagstuhl -- Leibniz-Zentrum
  f{\"u}r Informatik.

\bibitem[Bro41]{brooks_1941}
Rowland~Leonard Brooks.
\newblock On colouring the nodes of a network.
\newblock {\em Mathematical Proceedings of the Cambridge Philosophical
  Society}, 37(2):194--197, 1941.

\bibitem[BZ22]{BosekZ22}
Bartlomiej Bosek and Anna Zych{-}Pawlewicz.
\newblock Dynamic coloring of unit interval graphs with limited recourse
  budget.
\newblock In Shiri Chechik, Gonzalo Navarro, Eva Rotenberg, and Grzegorz
  Herman, editors, {\em 30th Annual European Symposium on Algorithms, {ESA}
  2022, September 5-9, 2022, Berlin/Potsdam, Germany}, volume 244 of {\em
  LIPIcs}, pages 25:1--25:14. Schloss Dagstuhl - Leibniz-Zentrum f{\"{u}}r
  Informatik, 2022.

\bibitem[CHvdH{\etalchar{+}}22]{AOOA-CHHRS}
Aleksander B.~G. Christiansen, Jacob Holm, Ivor van~der Hoog, Eva Rotenberg,
  and Chris Schwiegelshohn.
\newblock Adaptive out-orientations with applications, 2022.
\newblock https://arxiv.org/abs/2209.14087.

\bibitem[CLP18]{ChangLP18}
Yi{-}Jun Chang, Wenzheng Li, and Seth Pettie.
\newblock An optimal distributed ({\(\Delta\)}+1)-coloring algorithm?
\newblock In Ilias Diakonikolas, David Kempe, and Monika Henzinger, editors,
  {\em Proceedings of the 50th Annual {ACM} {SIGACT} Symposium on Theory of
  Computing, {STOC} 2018, Los Angeles, CA, USA, June 25-29, 2018}, pages
  445--456. {ACM}, 2018.

\bibitem[CR22]{DBLP:conf/icalp/CR22}
Aleksander B.~G. Christiansen and Eva Rotenberg.
\newblock Fully-dynamic {\(\alpha\)} + 2 arboricity decompositions and implicit
  colouring.
\newblock In {\em {ICALP}}, volume 229 of {\em LIPIcs}, pages 42:1--42:20.
  Schloss Dagstuhl - Leibniz-Zentrum f{\"{u}}r Informatik, 2022.

\bibitem[CS88]{DBLP:journals/ita/ChrobakS88}
Marek Chrobak and Maciej Slusarek.
\newblock On some packing problem related to dynamic storage allocation.
\newblock {\em {RAIRO} Theor. Informatics Appl.}, 22(4):487--499, 1988.

\bibitem[EFF85]{erdos}
Paul Erd{\"o}s, P{\'e}ter Frankl, and Zolt{\'a}n F{\"u}redi.
\newblock Families of finite sets in which no set is covered by the union of
  $r$ others.
\newblock In {\em Isreal Journal of Mathematics}, volume~51, pages 79--89,
  1985.

\bibitem[FHK16]{FraigniaudHK16}
Pierre Fraigniaud, Marc Heinrich, and Adrian Kosowski.
\newblock Local conflict coloring.
\newblock In Irit Dinur, editor, {\em {IEEE} 57th Annual Symposium on
  Foundations of Computer Science, {FOCS} 2016, 9-11 October 2016, Hyatt
  Regency, New Brunswick, New Jersey, {USA}}, pages 625--634. {IEEE} Computer
  Society, 2016.

\bibitem[Gab95]{Gabow1995}
Harold~N. Gabow.
\newblock Algorithms for graphic polymatroids and parametric s-sets.
\newblock In {\em Proceedings of the Sixth Annual ACM-SIAM Symposium on
  Discrete Algorithms}, SODA '95, pages 88--97, USA, 1995. Society for
  Industrial and Applied Mathematics.

\bibitem[GL17]{DBLP:conf/wdag/GhaffariL17}
Mohsen Ghaffari and Christiana Lymouri.
\newblock Simple and near-optimal distributed coloring for sparse graphs.
\newblock In Andr{\'{e}}a~W. Richa, editor, {\em 31st International Symposium
  on Distributed Computing, {DISC} 2017, October 16-20, 2017, Vienna, Austria},
  volume~91 of {\em LIPIcs}, pages 20:1--20:14. Schloss Dagstuhl -
  Leibniz-Zentrum f{\"{u}}r Informatik, 2017.

\bibitem[GW88]{GabowWestermann}
Harold Gabow and Herbert Westermann.
\newblock Forests, frames, and games: Algorithms for matroid sums and
  applications.
\newblock In {\em Proceedings of the Twentieth Annual ACM Symposium on Theory
  of Computing}, STOC '88, pages 407--421, New York, NY, USA, 1988. Association
  for Computing Machinery.

\bibitem[HdLT98]{HolmLT98}
Jacob Holm, Kristian de~Lichtenberg, and Mikkel Thorup.
\newblock Poly-logarithmic deterministic fully-dynamic algorithms for
  connectivity, minimum spanning tree, 2-edge, and biconnectivity.
\newblock In Jeffrey~Scott Vitter, editor, {\em Proceedings of the Thirtieth
  Annual {ACM} Symposium on the Theory of Computing, Dallas, Texas, USA, May
  23-26, 1998}, pages 79--89. {ACM}, 1998.

\bibitem[HKMT21]{HalldorssonKMT21}
Magn{\'{u}}s~M. Halld{\'{o}}rsson, Fabian Kuhn, Yannic Maus, and Tigran
  Tonoyan.
\newblock Efficient randomized distributed coloring in {CONGEST}.
\newblock In Samir Khuller and Virginia~Vassilevska Williams, editors, {\em
  {STOC} '21: 53rd Annual {ACM} {SIGACT} Symposium on Theory of Computing,
  Virtual Event, Italy, June 21-25, 2021}, pages 1180--1193. {ACM}, 2021.

\bibitem[HNW20]{DBLP:journals/corr/abs-2002-10142}
Monika Henzinger, Stefan Neumann, and Andreas Wiese.
\newblock Explicit and implicit dynamic coloring of graphs with bounded
  arboricity.
\newblock {\em CoRR}, abs/2002.10142, 2020.

\bibitem[HP20]{DBLP:conf/stacs/Henzinger020}
Monika Henzinger and Pan Peng.
\newblock Constant-time dynamic ({\(\Delta\)}+1)-coloring.
\newblock In Christophe Paul and Markus Bl{\"{a}}ser, editors, {\em 37th
  International Symposium on Theoretical Aspects of Computer Science, {STACS}
  2020, March 10-13, 2020, Montpellier, France}, volume 154 of {\em LIPIcs},
  pages 53:1--53:18. Schloss Dagstuhl - Leibniz-Zentrum f{\"{u}}r Informatik,
  2020.

\bibitem[HP22]{DBLP:journals/talg/HenzingerP22}
Monika Henzinger and Pan Peng.
\newblock Constant-time dynamic ({\(\Delta\)}+1)-coloring.
\newblock {\em {ACM} Trans. Algorithms}, 18(2):16:1--16:21, 2022.

\bibitem[HR20a]{HolmR20stoc}
Jacob Holm and Eva Rotenberg.
\newblock Fully-dynamic planarity testing in polylogarithmic time.
\newblock In Konstantin Makarychev, Yury Makarychev, Madhur Tulsiani, Gautam
  Kamath, and Julia Chuzhoy, editors, {\em Proccedings of the 52nd Annual {ACM}
  {SIGACT} Symposium on Theory of Computing, {STOC} 2020, Chicago, IL, USA,
  June 22-26, 2020}, pages 167--180. {ACM}, 2020.

\bibitem[HR20b]{HolmR20soda}
Jacob Holm and Eva Rotenberg.
\newblock Worst-case polylog incremental {SPQR}-trees: Embeddings, planarity,
  and triconnectivity.
\newblock In Shuchi Chawla, editor, {\em Proceedings of the 2020 {ACM-SIAM}
  Symposium on Discrete Algorithms, {SODA} 2020, Salt Lake City, UT, USA,
  January 5-8, 2020}, pages 2378--2397. {SIAM}, 2020.

\bibitem[HRT18]{HolmRT18}
Jacob Holm, Eva Rotenberg, and Mikkel Thorup.
\newblock Dynamic bridge-finding in {$\widetilde{O}(\log ^2 n)$} amortized
  time.
\newblock In Artur Czumaj, editor, {\em Proceedings of the Twenty-Ninth Annual
  {ACM-SIAM} Symposium on Discrete Algorithms, {SODA} 2018, New Orleans, LA,
  USA, January 7-10, 2018}, pages 35--52. {SIAM}, 2018.

\bibitem[HSS16]{HarrisSS16}
David~G. Harris, Johannes Schneider, and Hsin{-}Hao Su.
\newblock Distributed ({$\Delta$}+1)-coloring in sublogarithmic rounds.
\newblock In Daniel Wichs and Yishay Mansour, editors, {\em Proceedings of the
  48th Annual {ACM} {SIGACT} Symposium on Theory of Computing, {STOC} 2016,
  Cambridge, MA, USA, June 18-21, 2016}, pages 465--478. {ACM}, 2016.

\bibitem[HTZ14]{He2014OrientingDG}
Meng He, Ganggui Tang, and Norbert Zeh.
\newblock Orienting dynamic graphs, with applications to maximal matchings and
  adjacency queries.
\newblock In {\em ISAAC}, 2014.

\bibitem[KKPS14]{kopelowitz2013orienting}
Tsvi Kopelowitz, Robert Krauthgamer, Ely Porat, and Shay Solomon.
\newblock Orienting fully dynamic graphs with worst-case time bounds.
\newblock In {\em Automata, Languages, and Programming - 41st International
  Colloquium, {ICALP} 2014, Copenhagen, Denmark, July 8-11, 2014, Proceedings,
  Part {II}}, volume 8573 of {\em Lecture Notes in Computer Science}, pages
  532--543. Springer, 2014.

\bibitem[Kow07]{10.1016/j.ipl.2006.12.006}
\L{}ukasz Kowalik.
\newblock Adjacency queries in dynamic sparse graphs.
\newblock {\em Inf. Process. Lett.}, 102(5):191–195, May 2007.

\bibitem[KP06]{DBLP:conf/icalp/KhotP06}
Subhash Khot and Ashok~Kumar Ponnuswami.
\newblock Better inapproximability results for maxclique, chromatic number and
  min-3lin-deletion.
\newblock In {\em Automata, Languages and Programming, 33rd International
  Colloquium, {ICALP} 2006, Venice, Italy, July 10-14, 2006, Proceedings, Part
  {I}}, volume 4051 of {\em Lecture Notes in Computer Science}, pages 226--237.
  Springer, 2006.

\bibitem[KT81]{KiersteadTrotter1981extremal}
Henry~A Kierstead and William~T Trotter.
\newblock An extremal problem in recursive combinatorics.
\newblock {\em Congressus Numerantium}, 33(143-153):98, 1981.

\bibitem[Lin92]{DBLP:journals/siamcomp/Linial92}
Nathan Linial.
\newblock Locality in distributed graph algorithms.
\newblock {\em {SIAM} J. Comput.}, 21(1):193--201, 1992.

\bibitem[Mau21]{Maus21}
Yannic Maus.
\newblock Distributed graph coloring made easy.
\newblock In Kunal Agrawal and Yossi Azar, editors, {\em {SPAA} '21: 33rd {ACM}
  Symposium on Parallelism in Algorithms and Architectures, Virtual Event, USA,
  6-8 July, 2021}, pages 362--372. {ACM}, 2021.

\bibitem[MT20]{MausT20}
Yannic Maus and Tigran Tonoyan.
\newblock Local conflict coloring revisited: Linial for lists.
\newblock In Hagit Attiya, editor, {\em 34th International Symposium on
  Distributed Computing, {DISC} 2020, October 12-16, 2020, Virtual Conference},
  volume 179 of {\em LIPIcs}, pages 16:1--16:18. Schloss Dagstuhl -
  Leibniz-Zentrum f{\"{u}}r Informatik, 2020.

\bibitem[Pou94]{Poutre94}
Johannes A.~La Poutr{\'{e}}.
\newblock Alpha-algorithms for incremental planarity testing (preliminary
  version).
\newblock In Frank~Thomson Leighton and Michael~T. Goodrich, editors, {\em
  Proceedings of the Twenty-Sixth Annual {ACM} Symposium on Theory of
  Computing, 23-25 May 1994, Montr{\'{e}}al, Qu{\'{e}}bec, Canada}, pages
  706--715. {ACM}, 1994.

\bibitem[PR07]{DBLP:journals/tcs/ParnasR07}
Michal Parnas and Dana Ron.
\newblock Approximating the minimum vertex cover in sublinear time and a
  connection to distributed algorithms.
\newblock {\em Theor. Comput. Sci.}, 381(1-3):183--196, 2007.

\bibitem[RS20]{DBLP:conf/podc/RosenbaumS20}
Will Rosenbaum and Jukka Suomela.
\newblock Seeing far vs. seeing wide: Volume complexity of local graph
  problems.
\newblock In Yuval Emek and Christian Cachin, editors, {\em {PODC} '20: {ACM}
  Symposium on Principles of Distributed Computing, Virtual Event, Italy,
  August 3-7, 2020}, pages 89--98. {ACM}, 2020.

\bibitem[RTVX11]{DBLP:conf/innovations/RubinfeldTVX11}
Ronitt Rubinfeld, Gil Tamir, Shai Vardi, and Ning Xie.
\newblock Fast local computation algorithms.
\newblock In Bernard Chazelle, editor, {\em Innovations in Computer Science -
  {ICS} 2011, Tsinghua University, Beijing, China, January 7-9, 2011.
  Proceedings}, pages 223--238. Tsinghua University Press, 2011.

\bibitem[SV93]{SzegedyV93}
Mario Szegedy and Sundar Vishwanathan.
\newblock Locality based graph coloring.
\newblock In S.~Rao Kosaraju, David~S. Johnson, and Alok Aggarwal, editors,
  {\em Proceedings of the Twenty-Fifth Annual {ACM} Symposium on Theory of
  Computing, May 16-18, 1993, San Diego, CA, {USA}}, pages 201--207. {ACM},
  1993.

\bibitem[SW20]{DBLP:conf/stoc/SawlaniW20}
Saurabh Sawlani and Junxing Wang.
\newblock Near-optimal fully dynamic densest subgraph.
\newblock In {\em {STOC}}, pages 181--193. {ACM}, 2020.

\bibitem[Tho00]{Thorup00}
Mikkel Thorup.
\newblock Near-optimal fully-dynamic graph connectivity.
\newblock In F.~Frances Yao and Eugene~M. Luks, editors, {\em Proceedings of
  the Thirty-Second Annual {ACM} Symposium on Theory of Computing, May 21-23,
  2000, Portland, OR, {USA}}, pages 343--350. {ACM}, 2000.

\bibitem[Zuc06]{10.1145/1132516.1132612}
David Zuckerman.
\newblock Linear degree extractors and the inapproximability of max clique and
  chromatic number.
\newblock In {\em Proceedings of the Thirty-Eighth Annual ACM Symposium on
  Theory of Computing}, STOC '06, pages 681--690, New York, NY, USA, 2006.
  Association for Computing Machinery.

\end{thebibliography}

\end{document}